\documentclass[11pt,a4paper]{article}
\usepackage[sc,noBBpl]{mathpazo}
\usepackage{amsmath,amsfonts,amssymb,amsthm}
\usepackage[mathscr]{eucal}

\newcommand\mvector{\boldsymbol}

\newcommand\vd{\mvector{d}}

\newcommand\vp{\mvector{p}}
\newcommand\vq{\mvector{q}}

\newcommand\field{\mathbb}

\newcommand\R{\field{R}}

\newcommand\C{\field{C}}
\newcommand\Z{\field{Z}}

\newcommand\N{\field{N}}
\newcommand\Q{\field{Q}}

\renewcommand\Re{\operatorname{Re}}

\newcommand\tr{\operatorname{Tr}}

\newcommand\id{\operatorname{\mathrm{Id}}}

\newcommand\rmi{\mathrm{i}\mspace{1mu}}

\newcommand\mtext[1]{\quad\text{#1}\quad}

\def\nn{\nonumber }
\theoremstyle{plain}
\newtheorem{theorem}{Theorem}
\newtheorem{lemma}[theorem]{Lemma}

\newtheorem{conjecture}[theorem]{Conjecture}

\newtheoremstyle{note}{\topsep}{\topsep}{\slshape}{}{\scshape}{}{ }{}
\theoremstyle{note}

\newtheorem{remark}[theorem]{Remark}
\numberwithin{equation}{section}
\numberwithin{theorem}{section}
\usepackage[usenames]{color}
\begin{document}

\title{On algebraic construction of certain integrable and
  super-integrable systems}
\author{A. J.~Maciejewski$^1$,  M. Przybylska$^{2}$, A. V.~Tsiganov$^3$ \\
  \\
  \begin{minipage}{0.8\textwidth}
    \small $^1$J.~Kepler Institute of Astronomy, University of Zielona
    G\'ora, Licealna 9, \\ \quad PL-65--417 Zielona G\'ora, Poland
    (e-mail: maciejka@astro.ia.uz.zgora.pl)\\[0.5em]
    $^2$Institute of Physics, University of Zielona G\'ora,
    Licealna 9, \\\quad PL-65--417 Zielona G\'ora, Poland\\
    (e-mail:M.Przybylska@proton.if.uz.pl)\\[0.5em]
    $^3$St.~Petersburg State University, St.~Petersburg, Russia\\ 
    (e-mail: andrey.tsiganov@gmail.com)
  \end{minipage}} \date{}
\maketitle

\begin{abstract}
  We propose a new construction of two-dimensional natural
  bi-Hamiltonian systems  associated with a very simple Lie algebra.  The
  presented construction allows us to distinguish three families of
  super-integrable monomial potentials for which one additional first
  integral is quadratic, and the second one can be of arbitrarily high
  degree with respect to the momenta.  Many integrable systems with
  additional integrals of degree greater than two in momenta are
  given.  Moreover, an example of a  super-integrable system with first
  integrals of degree two, four and six in the momenta is found.
\end{abstract}

\section{Introduction}
The main aim of this paper is to study natural integrable systems
\begin{equation}
  \label{ham1} H_1=2p_1p_2+V(q_1,q_2),
\end{equation}
with two degrees of freedom associated with the various
representations of Lie algebra
\begin{equation}
  \label{alg} [N,a_+]=\kappa_1\, a_+,\qquad
  [N,a_-]=-\kappa_2a_-,\qquad
  [a_+,a_-]=0\,
\end{equation}
labelled by parameters $\kappa_{1}, \kappa_{2}\in\R$, and having the
Casimir element
\[
C=a_+^{\kappa_2} a_-^{\kappa_1}.
\]
An obvious representation of this algebra is the algebra of smooth
function defined on the phase space generated by
\begin{equation}
  \label{real-os}
  a_+=q_1, \qquad a_-=q_2, \qquad N=-\kappa_1 q_1
  p_1 +\kappa_2 q_2 p_2,
\end{equation}
with bracket $[a,b]:=\{a,b\}$, where $\{\cdot,\cdot\}$ denotes the
canonical Poisson bracket in $\R^4$.  

We met with this algebra investigating bi-Hamiltonian structures of
two-dimensional natural Hamiltonian systems with homogeneous
potentials and Newton's equations with homogeneous velocity
independent forces. Namely, let us consider a natural system given by
the Hamilton function of the form~\eqref{ham1} with a monomial
potential. More precisely, let us consider system governed by
Hamiltonian of the form
\begin{equation}
  \label{eq:hs}
  H\equiv a_+=2p_1p_2+q_1^aq_2^b,
\end{equation}
where $a$ and $b$ are real parameters.  For such a system the 
function
\begin{equation}
  \label{N-2} N=\alpha p_1q_1+\beta
  p_2q_2,
\end{equation}
satisfies the following equality
\[
\{N,a_+\}=(\alpha+\beta)\,a_+,
\]
provided that $\alpha$ and $\beta$ are fulfils equation
\[
a\,\alpha+b\,\beta+\alpha+\beta=0.
\]
If we can find the remaining generator $a_-=H_2$, then we have
integrable in the Liouville sense dynamical system with homogeneous
potential.  In order to get this additional first integral we use the
machinery of bi-Hamiltonian geometry, see \cite{ts06,ts10}.

It is worth to notice that the existence of function $N$ with the
prescribed property is related with a certain symmetry of the
system. Let $X_{F}$ denote the Hamiltonian vector field generated by
$F$.  Assume that for Hamiltonian~\eqref{eq:hs} there exists function
$N$ such that $\{N,H\}=c H$, where $c$ is a constant.  Then the
Hamiltonian vector field $X_N$ is a master symmetry of $X_H$, as
$[X_N,X_H]=cX_H$. Sometimes a master symmetry is called a conformal
symmetry, see \cite{Bogoyavlenskii:98::a}.

Let we assume that $X_N$ with $N$ of the following form
\begin{equation}
  \label{eq:n}
  N:=A(q_1,q_2)p_1 + B(q_1,q_2)p_2,
\end{equation}
where $A(q_1,q_2)$ and $B(q_1,q_2)$ are differentiable functions, is a
conformal symmetry of $X_H$.  Then it is easy to show that necessarily
$A=\alpha q_1$ and $B=\beta q_2$.  Thus, equality $\{N,H\}=c H$
implies that
\begin{equation}
  \label{eq:e}
  c=\alpha+\beta, \quad   \alpha a + \beta b +c =0.
\end{equation}
So, we recovered the assumed form of $N$, see~\eqref{N-2}. Now, the
last relation in~\eqref{alg} tells that we need an additional integral
$H_2\equiv a_{-}$.  Moreover, we require that $\{N,H_2\}={d}H_2$, for
a certain $d\in\R$, and this is equivalent that $X_N$ is also a
conformal symmetry of $X_{H_2}$.

We are going to work with systems~\eqref{eq:hs}  with two degrees of freedom because
such systems appeared as subsystems on invariant manifold of
$n$-dimensional Hamiltonian systems \cite{amp09,mp09}.

The plan of this paper is following. In the next section
we start with short description how Hamiltonian systems of the
form~\eqref{eq:hs} appeared in our investigations of the integrability
of natural Hamiltonian systems with homogeneous potentials.  In
Section~\ref{sec:irr} bi-Hamiltonian irregular Poisson manifolds as
well as their application for construction of first integrals of
considered systems are presented. The remaining sections contain
results of the integrability analysis. In Section~\ref{sec:bi} four
families of integrable systems with additional first integrals
quadratic in the momenta are given. In the next four sections
super-integrable cases in these families are distinguished. In
Section~\ref{sec:hi} examples of integrable systems with  additional
first integrals of degree greater than two in the momenta are given.
In Appendix we collected basic facts concerning the Gauss
hypergeometric equation which are used in this paper.

\section{Monomial potentials}
\label{sec:mon}
Let us consider Hamiltonian systems with $n$ degrees of freedom with
the Hamiltonian of the following classical form
\begin{equation}
  H=\frac{1}{2}\sum_{i=1}^n p_i^2+V(\vq),
  \label{eq:homo}
\end{equation}
where $\vq=(q_1,\ldots, q_n)$ and $\vp=(p_1,\ldots, p_n)$ are the
canonical coordinates and momenta; the potential $V(\vq)$ is a
homogeneous function of degree $k\in\Q$.  The strongest necessary
integrability conditions for such systems with $k\in \Z$ were obtained
thanks to an application of differential Galois methods, see \cite{Morales:99::c,Morales:01::a,mp:05::c}.
To derive conditions of this type one needs to know a particular
non-equilibrium solution of the considered systems.  For the systems
given by Hamiltonian~\eqref{eq:homo} a particular solution can be find
in systematic way.  Namely, if a non-zero $\vd\in\C^n$ satisfies
\begin{equation}
  V'(\vd)=\gamma\vd,
  \label{eq:darbcie}
\end{equation}
for some $\gamma\in\C$, then the system has a particular solution of
the following form
\begin{equation}
  \label{eq:ps}
  \vq(t)=\varphi(t)\vd , \qquad \vp(t)=\dot\varphi(t)\vd,
\end{equation}
where $\varphi(t)$ is a scalar function which is defined in the
following way. If vector $\vd$ satisfies equation~\eqref{eq:darbcie}
with $\gamma=0$, then $\varphi(t):=t$, and $\vd$ is called an improper
Darboux point of potential $V$. In this case, if $k\in\N$, and the
system is integrable in the Liouville sense, then all the eigenvalues
of the Hessian matrix $V(\vd)$ vanish, see~\cite{mp09}.

If the considered $\vd$ satisfies equation~\eqref{eq:darbcie} with
$\gamma\neq0$, then $\varphi(t)$ is a solution of equation $\ddot
\varphi=-\varphi^{k-1}$. In this case $\vd$ is called a proper Darboux
point of $V$, and if the system is integrable, then for a given
$k\in\Z\setminus\{-2,2\}$, all eigenvalues of $V''(\vd)$ belong to a
certain infinite subset of rational numbers $\Q$,
see~\cite{Morales:99::c,Morales:01::a,amp09}.

As we can see, in spite of the fact that the differential Galois
theory is quite involved the final result has the form of simple
arithmetic restrictions on the eigenvalues of matrix $V''(\vd)$.
Moreover, for polynomial potentials some relations between eigenvalues
of the Hessian calculated at all proper Darboux points exist. These
relations together with arithmetic restrictions on the eigenvalues
forced by necessary integrability conditions enable to find
effectively explicit forms of integrable potentials at least for small
$n$ and $k$.  Such a systematic analysis was initiated for $n=2$ in
\cite{Maciejewski:04::g, mp:05::c} and later on it was developed for
$n>2$, see~\cite{mp09,mp09a}.

However, for some classes of potentials the above approach does not
work. In particular it is a case if the consider potential does not
have any proper Darboux point, and moreover, the Hessian matrix
$V''(\vd)$ at each improper Darboux point is nilpotent.  For $n=2$ an
almost complete characterisation of such polynomial potentials is
given by the following lemma.
\begin{lemma}
  \label{lem:nodarboux}
  If a polynomial potential $V$ of degree $k>2$ does not have any
  proper Darboux point, then it is either equivalent to the following
  one
  \begin{equation}
    \label{eq:vkl}
    V_{k,l} = \alpha (q_2-\rmi q_1  )^{k-l}(q_2+\rmi q_1 )^{l} , \qquad
    \quad \alpha\in\C^\star:=\C\setminus\{0\}.
  \end{equation}
  for some $l=2,\ldots,k-2$, or $k=2s$ and $V$ has factor $(q_2\pm
  \rmi q_1)$ with multiplicity $s$.
\end{lemma}
We say that potential $V$ is equivalent to $W$ iff $V(\vq)=W(A\vq)$,
where $A$ is $n\times n$ matrix satisfying $AA^T=\beta\id_n$ for a
certain $\beta\in\C^{\star}$.

Let us notice that potential $V_{k,l}$ with $2\leq l\leq k-2$ has two
improper Darboux points $\vd_1=(1,\rmi)$ and
$\vd_2=(1,-\rmi)$. Moreover, $V''(\vd_1)$ and $V''(\vd_2)$ are
nilpotent.  In other words, for these potentials we do not have any
obstacles for the integrability.

The potentials of the form~\eqref{eq:vkl} look more attractive if we
introduce new canonical variables
\begin{equation}
  \label{eq:tozy}
  z_1=q_1+\rmi q_2,\qquad  z_2=q_1-\rmi q_2, \qquad
  y_1=\frac{1}{2}(p_1-\rmi p_2),\quad y_2=\frac{1}{2}(p_1+\rmi p_2).
\end{equation}
In these variables the Hamiltonian function has the form
\begin{equation}
  \label{eq:haniki}
  H=2y_1y_2 +z_1^lz_2^{k-l}.
\end{equation}
Clearly, the systems governed by Hamiltonian function of the form
\eqref{eq:hs} are just generalisation of~\eqref{eq:haniki} for cases
when $k$ and $l$ are rational.

\section{Irregular Poisson manifolds}
\label{sec:irr}
A bi-Hamiltonian manifold $M$ is a smooth manifold endowed with a pair
of compatible Poisson bi-vectors $P$ and $P'$ such that
\begin{equation}
  \label{m-eq1}
  [\![P,P']\!]=0,\qquad [\![P',P']\!]=0,\qquad
\end{equation}
where $[\![\cdot,\cdot]\!]$ is the Schouten bracket.
 If $\dim M=2n$,
and $P$ is invertible Poisson bivector on $M$, the Nijenhuis operator
which is called also the hereditary, or recursion operator, is defined
as
\[\mathcal N=P' P^{-1}.
\]

We say that a bi-Hamiltonian manifold $M$ is regular iff $\dim M=2n$,
$P$ is invertible Poisson bivector, and the recursion operator
$\mathcal N$ has, at every point, $n$ distinct functionally
independent eigenvalues. If the recursion operator $\mathcal N$ does
not have this property then we say that bi-Hamiltonian manifold $M$ is \textit {irregular}.

For a regular  bi-Hamiltonian manifold $M$ functions
\begin{equation}
  \label{aux-int}
 H_k=\frac1{2k}\tr\mathcal N^k, \mtext{for} k\in\N,
\end{equation}
form a bi-Hamiltonian hierarchy on $M$, i.e., the Lenard relations
hold
\begin{equation}
  \label{lenard}
  P' \mathrm{d}{ H}_k=P \mathrm{d}{ H}_{k+1}\,,\quad\text{for all}\quad k\ge 1.
\end{equation}
Generally it is unknown    if it is possible  to construct in a
systematic way a full involutive  set  of functions on irregular
bi-Hamiltonian manifolds.

Anyway,  we are going to apply the
algebra \eqref{alg} to construct  integrable systems on
two-dimensional irregular bi-Hamiltonian manifolds.

\subsection{The rational Calogero-Moser system}
In order to justify better some aspects of our construction we
consider the $n$-particle rational Calogero-Moser model associated
with the root system $\mathcal A_n$. It is defined by the Hamilton
function
\begin{equation}
  H=\dfrac{1}{2}\sum_{i=1}^{n}p_{i}^{\;2}-{a^2}\sum_{i\neq
    j}^n \frac{1}{(q_{i}-q_{j})^{2}},  \label{ham-cal}
\end{equation}
where $a$ is a coupling constant.  The canonical variables $(\vq,\vp)$
satisfy the standard Poisson bracket relations
$\{q_{i},p_{j}\}=\delta_{ij}$, and the associated Poisson bivector
will be denoted by $P$.

This system admits independent and commuting first integrals
\begin{equation*}
  H_k=\dfrac{1}{k!}\,\tr\, L^k, \mtext{for} k=1,\ldots,n,
\end{equation*}
where $ L$ is the standard Lax matrix
\[
 L
=\begin{bmatrix} p_1 & \frac{a}{q_{1}-q_{2}} & \cdots & \frac{a}{q_{1}-q_{n}}\\
  \frac{a}{q_{2}-q_{1}} & p_2 & \ddots & \vdots\\
  \vdots & \ddots & \ddots & \vdots\\
  \frac{a}{q_{n}-q_{1}} & \frac{a}{q_{n}-q_{2}} & \cdots & p_n
\end{bmatrix} .
\]
Besides these $n$ integrals of motion the rational Calogero-Moser
system admits $(n-1)$ additional functionally independent integrals of
motion $K_{m}$ given by
\[
K_{m}= m g_{1}H_{m}- g_{m} H_{1}, \qquad m=2,...,n,
\]
where
\begin{equation*}
  g_j = \frac{1}{2} \,
  \left\{\sum_{i=1}^{n} q_{i}^{\;2}, H_{j} \right\}, \mtext{for}
  j=1,...,n.
\end{equation*}
According to \cite{ts10} integrals of motion $H_k$ and $K_m$ can be
obtained from the Hamilton function $H=H_2$ given by \eqref{ham-cal}
as polynomial in momenta solutions of the following equations
\begin{equation}
  \label{cal-int}
  P\,\mathrm{d}H=k^{-1}\,P'\,\mathrm{d}\ln H_k=
(m-1)^{-1}\,P'\,\mathrm{d}\ln K_m,
\end{equation}
where $P'$ is the  Poisson bivector compatible with
$P$
given by
\begin{equation*}
  P'=\begin{bmatrix}
    R & \Pi \\
    -\Pi & 0
  \end{bmatrix}
  +
  \begin{bmatrix}
    0& \Lambda \\
    -\Lambda & M
  \end{bmatrix},
\end{equation*}
and $n\times n$ matrices $\Pi$, $\Lambda$, $R$ and $M$ have the
following entries
\[
\Pi_{ij}=p_ip_j,\qquad \Lambda_{ij}=q_i\sum_{k\neq j}^n
\dfrac{a^2}{(q_j-q_k)^3},
\]
and
\begin{equation*}
  R_{ij}=\sum_{k=1}^n\left(\dfrac{\partial \Pi_{jk}}{\partial
      p_i}-\dfrac{\partial \Pi_{ik}}{\partial p_j}\right)q_k, \quad
  M_{ij}= \sum_{k=1}^n\left(\frac{\partial \Lambda_{ki}}{\partial
      q_j}-
\frac{\partial \Lambda_{kj}}{\partial q_i}\right)p_k.
\end{equation*}
In this case recursion operator $\mathcal N=P'P^{-1}$ generates only
the Hamilton function
\begin{equation}
  \label{cal-deg}
  \tr\,\mathcal N^k=2(2H)^{k}\,\quad\text{such that}\quad
 P\,\mathrm{d}H=P'\mathrm{d}\ln H,
\end{equation}
and instead of the standard Lenard relations \eqref{lenard} we have
equations \eqref{cal-int}. Thus, the described Calogero-Moser system
is a super-integrable system on an irregular bi-Hamiltonian manifold
$\mathbb R^{2n}$.

\subsection{Algebraic construction of two-dimensional bi-Hamiltonian
  systems on irregular manifolds}
In order to get integrable systems associated with the Lie algebra
\eqref{alg} we start with the canonical Poisson bivector
\[
P=
\begin{bmatrix}
  0 & 0 & 1 & 0 \\
  0 & 0 & 0 & 1 \\
  -1 & 0 & 0 & 0 \\
  0 & -1 & 0 & 0
\end{bmatrix},
\]
and the Hamiltonian vector field
\[
X=PdH_1=\left[ 2p_2, 2p_1, -\dfrac{\partial V(q_1, q_2)}{\partial
    q_1}, -\dfrac{\partial V(q_1, q_2)}{\partial q_2} \right]^T .
\]
We suppose that this vector field $X$ is the bi-Hamiltonian vector
field with respect to a certain Poisson bivector $P'$, and some second
integral of motion $H_2$
\[
X=P\mathrm{d}H_1=P' \mathrm{d}H_2.
\]
The bivector $P'$ has to be compatible with canonical bivector $P$,
and it has to be the Poisson bivector. Hence, we have the following
relations
\begin{equation}
  \label{m-eq} [\![P,P']\!]=0,\qquad
  [\![P',P']\!]=0,\qquad\mbox{and}\qquad \{H_1,H_2\}=\{H_1,H_2\}'=0.
\end{equation}
Additionally we postulate that bivector $P'$ is a Lie derivative of
canonical bivector $P$ along the vector field $Y$ proportional to the
Hamiltonian vector field $X$
\begin{equation}\label{assum-1}
  P'=\mathcal L_Y
  P,\qquad Y=\rho X,
\end{equation}
where multiplier $\rho$ is equal to
\[
\rho=\alpha p_1q_1+\beta p_2q_2,
\]
and so it coincides with $N$ in \eqref{N-2}.

Assumption \eqref{assum-1} ensures that $P$ and $P'$ are compatible
bivectors, i.e., $[\![P,P']\!]=0$, and that eigenvalues of the
corresponding recursion operator $\mathcal N=P'P^{-1}$ are functions
which depend only on $H_1$. In fact, we have
\[
\det (\mathcal
N-\lambda\,\id_{4})=\lambda^2\Bigl(\lambda+(\alpha+\beta)H_1\Bigr)^2\,,
\]
see \eqref{cal-deg}. Thus, we have constructed an irregular bi-Hamiltonian
manifold $\R^4$.
\begin{theorem}
  \label{pro:bi}
  Bivector $P'$ \eqref{assum-1} satisfies the Jacobi condition
  \begin{equation}
    \label{poi} [\![P',P']\!]=0
\end{equation}
  iff
  \begin{equation}\label{1-pot}
    \alpha=0,\qquad V=\dfrac{ f(q_1)}{q_2},\qquad \rho= \beta
    p_2q_2
  \end{equation}
  or
  \begin{equation}\label{2-pot}
    V=\left(\gamma+
      g(q_1^{-\beta/\alpha}q_2)\right)\,q_1^{-(\alpha+\beta)/\alpha},\qquad
    \rho=\alpha p_1q_1+\beta p_2q_2\,,
  \end{equation}
  where $ f$ and $g$ are arbitrary smooth functions,  $\alpha\neq0$,
  $\beta$ and $\gamma$ are
  arbitrary numbers.
\end{theorem}
\begin{proof}
In order to prove this theorem we have to substitute the second
bivector
\begin{equation}\label{poi-2}
  P'=\mathcal L_Y P=
  \begin{bmatrix}
    0 &   -2\alpha p_1q_1+2\beta p_2q_2 &   -2\alpha p_1p_2+\alpha q_1\dfrac{\partial V}{\partial q_1} &   -2\beta p_2^2+\alpha q_1\dfrac{\partial V}{\partial q_2} \\
    \\
    * &   0 &   -2\alpha p_1^2+\beta q_2 \dfrac{\partial V}{\partial q_1} &   -2\beta p_1p_2+\beta q_2 \dfrac{\partial V}{\partial q_2}\\
    \\
    * & * &   0 &   -\alpha p_1 \dfrac{\partial V}{\partial q_2}+\beta p_2\dfrac{\partial V}{\partial q_1} \\
    \\
    * & * & * &   0 \\
  \end{bmatrix}
\end{equation}
into the Schouten bracket~\eqref{poi}, and to solve the resulting system of partial
differential equations.
\end{proof}

\begin{theorem}
  Hamilton function $H_1$ defined by \eqref{ham1} with potential $V$
  of the form \eqref{1-pot} or \eqref{2-pot} satisfies the following
  relations
  \begin{equation}\label{alg-g1}
    \{H_1,\rho\}=-(\alpha+\beta) H_1\,,
  \end{equation}
  and
  \begin{equation}\label{alg-g12}
    \{H_1,\rho\}'=\Bigl(-(\alpha+\beta) H_1\Bigr)^2\,,
  \end{equation}
  and
  \begin{equation}\label{log-1}
    X=P\mathrm{d}H_1=-(\alpha+\beta)^{-1} P'\mathrm{d}\ln H_1.
  \end{equation}

\end{theorem}
This therem is a direct consequence of the special form of $P'$ and
$H_1$.

Inspired by the equations \eqref{cal-int} for the Calogero-Moser
super-integrable system we suppose that our systems are bi-Hamiltonian
systems with respect to logarithm of the second integral of motion
$H_2$, i.e.,
\begin{equation}\label{assum-2}
  X=P\mathrm{d}H_1=\kappa_2^{-1}P'
  \mathrm{d} \ln H_2.
\end{equation}
This additional assumption is equivalent to the requirement that $H_2$
and $\rho$ satisfy supplementary relations
\begin{equation}\label{alg-g2}
  \{H_2,
  \rho\}=\kappa_2\,H_2\,,\qquad\mbox{and}\qquad
  \{H_2,\rho\}'=-(\alpha+\beta)\kappa_2\,H_1H_2\,.
\end{equation}
The obtained relations \eqref{alg-g1}, \eqref{alg-g12} and
\eqref{alg-g2} show that generators $H_1$, $H_2$ and $\rho$ form
linear and quadratic Poisson algebras with respect to the brackets
$\{\cdot,\cdot\}$ and $\{\cdot,\cdot\}'$, respectively.
\begin{remark}
  Assumptions \eqref{assum-1} and \eqref{assum-2} yield the
  representation of the algebra \eqref{alg}, but we can derive
  \eqref{assum-2} from the algebraic relations \eqref{alg} only if
  \eqref{assum-1} holds. Of course, form of \eqref{assum-2} is a
  little strange because the logarithm of an integral of motion is
  itself integral of motion.
\end{remark}

The equation \eqref{assum-2} is equivalent to the following system of
the partial differential equations
\begin{equation}
  \label{par-eq}
  \kappa_2 H_2
  \left[
    2p_2 ,
    2p_1 ,
    -\dfrac{\partial V}{\partial q_1} ,
    -\dfrac{\partial V}{\partial q_2}
  \right]^T
  =P'
  \left[
    \dfrac{\partial H_2}{\partial q_1},
    \dfrac{\partial H_2}{\partial q_2},
    \dfrac{\partial H_2}{\partial p_1},
    \dfrac{\partial H_2}{\partial p_2}
  \right]^T,
\end{equation}
with respect to unknown function $H_2=H_2(p_1,p_2,q_1,q_2)$,
parameters $\alpha,\beta,\kappa_2$ and functions $ f(q_1)$ or $
g(q_1^{-\beta/\alpha}q_2)$ which appear in definition of potential
$V(q_1,q_2)$ given by \eqref{1-pot}, or \eqref{2-pot}. The explicit
form of $P'$ appearing in~\eqref{par-eq} is given by~\eqref{poi-2}.
\begin{remark}
  Sometimes equations \eqref{par-eq} have different functionally
  independent solutions $H_2$ associated with the same potential
  $V$. In this case we obtain the so-called super-integrable system
  \cite{ts08a,ts08b}, and our algebra \eqref{alg} is some subalgebra
  of the complete polynomial algebra of integrals of motion.
\end{remark}
In order to get integral of motion $H_2$ we additionally assume that it is
polynomial in momenta, i.e., it has the form
\begin{equation}\label{anzatse}
  H_2=\sum_{k,m=0}^M
  g_{km}(q_1,q_2)p_1^m p_2^k,
\end{equation}
where $g_{km}(q_1,q_2)$ are smooth functions. Substituting this
expression into~\eqref{par-eq} we obtain partial differential
equations for unknown functions $g_{km}$, $ f$, $g$,  and
parameters $\alpha, \beta$. We can solve them for a fixed value of $M$.
For instance, in the cubic case $M=3$, for $\alpha=0$ and
$\kappa_2=0$, we obtain the following solution
\[
\rho=-q_2p_2,\qquad V= \dfrac{(c_1q_1+c_2)^{-1}}{q_2},\qquad
H_2=H_1\left(\dfrac{c_1q_1+c_2}{c_1}\,p_1+\rho\right).
\]
Of course, the Hamiltonian
\[
H_1=2p_1p_2+\dfrac{(c_1q_1+c_2)^{-1}}{q_2},
\]
possesses the first integral
\[I_2=\dfrac{c_1q_1+c_2}{c_1}\,p_1-q_2p_2,
\]
which is linear in the momenta. However, this first integral $I_2$
does not satisfy \eqref{par-eq} in contrast with cubic integral $H_2$.

In the next step, if we look for a first integral of degree $M=4$ in the
momenta, then for $\alpha=0$, we reproduce previous solution and get
one new
\[
\rho=-p_2q_2,\qquad V=\dfrac{f}{q_2},\qquad
f=(c_1q_1^2+c_2q_1+c_3)^{-1/2},
\]
associated with the fourth order in momenta integral of motion
\begin{equation}\label{form-br}
  H_2=-\dfrac{2p_1\rho+
    f}{4}\,\left(c_1\rho^2+\dfrac{ fp_1^2+( 2p_1\rho+ f
      ) f'}{q_2 f^3}\right).
\end{equation}
The obtained potential $V$ is non-polynomial and non-homogeneous
function.  Below we will select from generic solutions certain
particular solutions which are given by homogeneous functions.

\section{Bi-Hamiltonian systems with second order integrals of motion}
\label{sec:bi}

Let us consider potential $V(q_1,q_2)$ given by \eqref{2-pot} with
\begin{equation}\label{rest-1}
  \gamma=0, \qquad\text{and}
  \qquad  g(z)=z^d.
\end{equation}
That is
\begin{equation*}
  V(q_1,q_2)=q_1^{-\frac{\beta(d+1)+\alpha}{\alpha}}q_2^d,
\end{equation*}
and so, in this case the Hamiltonian has the form
\begin{equation}\label{ham-mon}
  H_1=2p_1p_2+q_1^{-\frac{\beta(d+1)+\alpha}{\alpha}}\,q_2^d.
\end{equation}
Taking this Hamiltonian we look for solutions $H_2$ of equations
\eqref{par-eq}.  We say that a solution of these equations is trivial
if the corresponding $V(q_1,q_2)$ does not depend on both variables.
We do not distinguish solutions which are obtained one from the other
by a permutation of pairs of variables $(p_1,q_1)$ and $(p_2,q_2)$.

Substituting Hamilton function $H_1$ of the form \eqref{ham-mon} into
the equations \eqref{par-eq}, and taking  $H_{2}$ of the
form~\eqref{anzatse} with $M=2$, one can get the following result.
\begin{theorem}
  \label{prop12}
  If $\kappa_2=1$, $\alpha\neq0$ and $\beta\neq 0$ equations
  \eqref{par-eq} have only four solutions $V(q_1,q_2)$ related to
  polynomial integrals of motion of degree two in the momenta
  \begin{center}
    \begin{tabular}{|l|l|l|l|l|}
      \hline
      & & & &\\
      1&  $\alpha=-\beta $ &  $H_1=2p_1p_2+q_1^dq_2^d$ & $I_2=q_1p_1-q_2p_2$& $H_2=0$ \\
      & & & &\\
      2& $\alpha=-2\beta(d+1)$  &  $ H_1=2p_1p_2+\dfrac{q_2^d}{\sqrt{q_1}}$ & $I_2=2p_1(q_2p_2-p_1q_1)+\dfrac{q_2^{d+1}}{\sqrt{q_1}}$ & $H_2=I_2^{-1/\alpha}$ \\
      & & & &\\
      3& $\alpha= -\dfrac{\beta (d+1)}{2}$ &  $H_1=2p_1p_2+q_1q_2^d$ & $I_2=p_1^2+\dfrac{q_2^{d+1}}{d+1}$
      &$H_2=I_2^{-1/2\alpha}$ \\
      & & & &\\
      4& $\alpha=\beta$ &  $H_1=2p_1p_2+\dfrac{q_2^d}{q_1^{d+2}}$ & $I_2=(p_1q_1-p_2q_2)^2-\dfrac{2q_2^{d+1}}{q_1^{d+1}}$ &$H_2=I_2H_1^{-1/2\alpha}$ \\
      \hline
    \end{tabular}
  \end{center}
\end{theorem}
\vskip0.2truecm
\begin{remark}
  In the case 2 and 3 $H_2=I_2$ satisfies equations \eqref{par-eq}
  with $\kappa_2=-\alpha$, and $\kappa_2=-2\alpha$,
  respectively. Similarly, in case 4 $H_2=I_2^{2\alpha}/ H_1$
  satisfies equations \eqref{par-eq} with $\kappa_2=2\alpha$.
\end{remark}
For each case in the above proposition we suppose that  $I_1=H_1$ and $I_2$ are
the action variables, and $\omega_1$, $\omega_2$ are the corresponding
angle variables.   Of course, this must be checked case by case.

 In these variables the equations of motion have the form
\begin{equation*}
  \dot I_1=0,  \quad \dot I_2 =0 , \quad  \dot \omega_1 =1  \quad  \dot \omega_2=0.
\end{equation*}
Thus, $\omega_2$ is a first integral functionally independent with
$I_1$ and $I_2$ \cite{ts07f}. The algebraic relations \eqref{alg-g1}, \eqref{alg-g12}, and \eqref{alg-g2} allow us to express $\rho$ in term of the
action-angle variables. In fact, we have
\begin{equation}\label{rho-w}
  \begin{array}{lll}
    \mtext{for case 1:} & \{\rho, I_2\}=0,\qquad & \rho=I_2,  \\
    \\
    \mtext{for case 2:} & \{\rho, I_2\}=\alpha I_2,\qquad & \rho=-(\alpha+\beta)I_1\omega_1-\alpha I_2\omega_2+F(I_1, I_2),  \\
    \\
    \mtext{for case 3:} & \{\rho, I_2\}=2\alpha I_2 & \rho=-(\alpha+\beta)I_1\omega_1-2\alpha I_2\omega_2+F(I_1, I_2),\\ \\
    \mtext{for case 4:}&\{\rho, I_2\}=0 & \rho=-(\alpha+\beta)I_1\omega_1+F(I_1, I_2)\,.
  \end{array}
\end{equation}
Here angle variables $\omega_{k}$ are defined up to canonical
transformations $\omega_k\to \omega_k+f(I_k)$, for $k=1,2$.

Later we show that families of systems given by Theorem~\ref{prop12} are
super-integrable for specific values of parameter $d$.  That is, for
these values of $d$ the systems admit three functionally independent
first integrals $H_1$, $I_2$ and $H_3$.  Some examples of
super-integrable systems with additional integral $H_3$ of third,
fourth and sixth order in momenta may be found in \cite{ts07f,ts08a}.

According to \cite{ts07f,ts08a,ts08b}, additional integral of motion
$H_3$ is a function on the action variables $I_1$, $I_{2}$ and one
angle variable $\omega_2$.  Below we prove that if we know additional
polynomial integral of motion $H_3$, then, using complete algebra of
integrals of motion, we can get the angle variable $\omega_2$
algebraically.

Some examples of super-integrable systems with additional integral
$H_3$ of third, fourth and sixth order in momenta may be found in
\cite{ts07f,ts08a}.

\begin{remark}
  Usually, see, e.g., \cite{cd06,mw07,ts08a,ts08b}, in the theory of
  super-integrable systems we study polynomial Poisson algebra of
  integrals of motion.  In this paper we add some extra generator
  $\rho$ to this algebra. Let us recall that $\rho$ plays the central
  role in our bi-Hamiltonian construction \eqref{assum-1}.
\end{remark}

\section{Case 1. Radial potential}
\label{sec:c1}
In this section we consider Hamiltonian system corresponding to case 1
in Theorem~\ref{prop12} and in particular we look for values of
$d$ for which the system is super-integrable.

We introduce new canonical variables $(r, \varphi, p_r, p_{\varphi})$
defined by the following equations
\begin{equation}
  q_1=r(\cos\varphi-\rmi \sin\varphi)=re^{-\rmi \varphi},\qquad
  q_2=r(\cos\varphi+\rmi \sin\varphi)=re^{\rmi \varphi},
\end{equation}
and
\begin{equation}
  p_1=\frac{e^{\rmi\varphi}}{2}\left(p_r+\frac{\rmi}{r}p_{\varphi}\right),\qquad
  p_2=\frac{e^{-\rmi\varphi}}{2}\left(p_r-\frac{\rmi}{r}p_{\varphi}\right).
\end{equation}
Hamiltonian $H_1$ and first integral $J_2:=-\rmi I_2$ in new variables
read
\begin{equation}
  H_1=\frac{p_r^2}{2}+\frac{p_{\varphi}^2}{2r^2}+V(r),\quad
  V(r)=r^{2d},\quad  J_2 =p_{\varphi}.
  \label{eq:radio}
\end{equation}

Thus after this transformation we obtain Hamiltonian of the classical
problem of the motion in the field of the radial force written in
polar coordinates. In the Cartesian coordinates $x_1=r\cos\varphi$,
and $x_2=r\sin\varphi$, Hamiltonian $H_1$ and the first integral $J_2$
read
\begin{equation}
  H_1=\frac{1}{2}(y_1^2+y_2^2)+r^{2d},\qquad
  r^2=x_1^2+x_2^2,\qquad J_2=y_1x_2-y_2x_1,
  \label{eq:radstan}
\end{equation}
where $y_1$ and $y_2$ are the momenta conjugated with $x_1$ and
$x_2$. Thus we obtain a natural Hamiltonian system with the standard form
of the kinetic energy and the radial potential.  The problem of the
existence of one more functionally independent first integral of this
system is the question about maximal super-integrability. For natural
Hamiltonian systems with homogeneous potentials possessing a proper
Darboux point necessary conditions of the maximal super-integrability
with first integrals which are meromorphic functions of coordinates
and momenta were formulated in \cite{mp:08::c}. These conditions were
obtained from the analysis of invariants of the differential Galois
group of variational equations.  The considered radial potential
possesses infinitely many proper Darboux points and results of
\cite{mp:08::c} can be directly applied. In fact, in Section 3 of this
paper Hamiltonian given by $H_1$ in \eqref{eq:radstan} was considered
and the following result was proved.
\begin{theorem}
  The radial potential \eqref{eq:radio} is super-integrable iff
  $d=-1/2$ or $d=1$.
\end{theorem}
Both of distinguished cases are indeed super-integrable with the
additional first integrals
\begin{equation}
  H_3 =(p_1 - p_2) (p_1 q_1 - p_2 q_2) - \dfrac{q_1 + q_2}{2 \sqrt{q_1q_2}} ,
\end{equation}
for $d=-1/2$, and
\begin{equation}
  H_3 =2 p_2^2 + q_1^2 ,
\end{equation}
for $d =1$, respectively.  This result reminds us the classical
Bertrand theorem~\cite{05.0470.01} which states that the only radial
potentials $V(r)=\alpha r^k$, for which all bounded orbits are
periodic, are those with $k=-1$ and $k=2$. The condition that all
bounded orbits of a Hamiltonian system are periodic means that the
system is degenerated, i.e., all its invariant tori are one
dimensional.  Such degeneration appears if the system is maximally
super-integrable.

\section{Case 2}
\label{sec:c2}
In this section we consider Hamiltonian system corresponding to case 2
in Theorem~\ref{prop12}. In particular, we distinguish all values
of $d$ for which the system is super-integrable. We achieve this
thanks to partial separation of variables and a direct integration of
equations of motion. In \cite{Nakagawa:02::} this system is called
quasi-separable and in fact we were able to find only one separation
relation.
\subsection{Algebraic super-integrability}
One can easy observe that the syzygi exists between $H_1$ and $I_2$
\begin{equation}
  \label{eq:i2c2}
  I_2=-2p_1^2q_1+q_2H_1,
\end{equation}
but this is not yet a separation relation. We have to make the change of
independent variable $t\to \tau$ defined by
$\mathrm{d}\tau/\mathrm{d}t=1/\sqrt{q_1}$, that transforms our
Hamiltonian system into
\begin{equation}
  \begin{split}
    &q_1'=2p_2\sqrt{q_1},\qquad p_1'=\frac{1}{2}\frac{q_2^d}{q_1},\\
    & q_2'=2p_1\sqrt{q_1},\qquad p_2'=-dq_2^{d-1},
  \end{split}
\end{equation}
where ${}'$ denotes differentiation with respect to $\tau$. This
transformation changes momenta and now we obtain one separation
relation
\[
I_2=-\frac{1}{2}q_2'^2+q_2H_1.
\]
It gives
\[
\int\frac{\mathrm{d}q_2}{\sqrt{2(q_2H_1-I_2)}}=\tau+\beta_1,
\]
where $\beta_1$ is a constant of integration.  Here we chosen square
root with sign $+$, that correspond $p_1>0$. This integral is
elementary and we obtain
\begin{equation}
  H_1^{-1}\sqrt{2(q_2H_1-I_2)}=\tau +\beta_1.
  \label{eq:natau}
\end{equation}
In further calculations we choose $\beta_1=0$, and then we obtain
\begin{equation}
  q_2=\frac{2I_2+H_1^2\tau^2}{2H_1}.
  \label{eq:naq2}
\end{equation}
Next, we substitute this function into integral $H_1$, that is, into
equation
\[
H_1=\dfrac{1}{2q_1}\frac{\mathrm{d}q_1}{\mathrm{d}\tau}\frac{\mathrm{d}q_2}{\mathrm{d}\tau}+
\frac{q_2^d}{\sqrt{q_1}},
\]
and  obtain
\begin{equation}
  \frac{H_1\tau}{2q_1}\frac{\mathrm{d}q_1}{\mathrm{d}\tau}+\frac{(2I_2+H_1^2\tau^2)^d}{(2H_1)^d\sqrt{q_1}}=H_1.
  \label{eq:naqq1}
\end{equation}
Then, we introduce new dependent variable $v=v(x)$ defined by
\[
q_1=\left(\frac{I_2}{H_1}\right)^{2d}\dfrac{v^2}{H_1^2}.
\]
Making this substitution in equation \eqref{eq:naqq1} and dividing
it by $H_1$ we obtain
\begin{equation}
  \dfrac{\tau}{v}\dfrac{\mathrm{d} v}{\mathrm{d}\tau}+\dfrac{(2I_2+H_1^2\tau^2)^d}{(2I_2)^dv}=1.
  \label{eq:robol}
\end{equation}
Finally, we make the following change of independent variable
\[
\tau \longmapsto x=-\dfrac{H_1^2\tau^2}{2I_2},
\]
and from \eqref{eq:robol} we obtain
\begin{equation}
  2x\dfrac{\mathrm{d} v}{\mathrm{d}x}-v+(1-x)^d=0.
  \label{eq:fini}
\end{equation}
This is a non-homogeneous linear equation. Dividing it by $(1-x)^d $
and differentiating with respect to $x$ we obtain homogeneous equation
of the second order
\begin{equation}
  x(1-x)\dfrac{\mathrm{d}^2 v}{\mathrm{d}x^2}+
  \dfrac{1-(1-2d)x}{2}\dfrac{\mathrm{d} v}{\mathrm{d}x}-\dfrac{d}{2}v=0.
  \label{eq:hipek}
\end{equation}
In this equation we recognise immediately Gauss hypergeometric
differential equation
\begin{equation}
  x(1-x)w''+[c-(a+b+1)x]w'-abw=0,
  \label{eq:hipcio}
\end{equation}
with parameters
\[
a=-\dfrac{1}{2},\qquad b=-d,\qquad c=\dfrac{1}{2}.
\]
The general solution of equation \eqref{eq:hipek} is following
\begin{equation}
  \label{eq:vsol}
  v(x)=\beta_2\sqrt{x}+\beta_3\phantom{\vert}_2F_{1}\left(-\frac{1}{2},-d,
    \frac{1}{2},x\right),
\end{equation}
where $\beta_2$, $\beta_3$ are constants of integration, and
$\phantom{\vert}_2F_{1}(a,b,c;x)$ denotes the hypergeometric function
with parameters $a$, $b$ and $c$. Let us recall that
$\phantom{\vert}_2F_{1}(a,b,c;x)$ is a solution of the hypergeometric
differential equation \eqref{eq:hipcio} which is holomorphic at the
origin, see Appendix.

Substituting \eqref{eq:vsol} into non-homogeneous
equation~\eqref{eq:fini} we obtain
\begin{equation}
  \label{eq:cons}
  -(\beta_3-1)(1-x)^d=0.
\end{equation}
Hence, $\beta_3=1$, and solution of equation~\eqref{eq:fini} is given
by
\begin{equation}
  \label{eq:vsolnh}
  v(x)=\beta_2\sqrt{x}+\phantom{\vert}_2F_{1}\left(-\frac{1}{2},-d,
    \frac{1}{2},x\right).
\end{equation}
Next, we return to the original dependent and independent variables and
express $\tau$ as a function of $q_2$, see~\eqref{eq:naq2}. After this
transformation the solution of \eqref{eq:naqq1} takes the following form
\begin{equation}
  \label{eq:q1}
  q_1=\frac{1}{H_1^2}\left[\beta_2\sqrt{2(H_1q_2-I_2)}+
    \left(\frac{I_2}{H_1}\right)^d \phantom{\vert}_2F_{1}\left(-\frac{1}{2},-d,
      \frac{1}{2},1-\frac{H_1}{I_2}q_2\right)\right]^2.
\end{equation}
From the obtained expression one can easily find  the following
first integral
\begin{equation}
  H_3=\sqrt{2}\beta_2=
\frac{\sqrt{q_1}\,H_1-\left(\frac{I_2}{H_1}\right)^d
\phantom{\vert}_2F_{1}\left(-\frac{1}{2},-d,
      \frac{1}{2},1-\frac{H_1}{I_2}q_2\right)}{\sqrt{q_2H_1-I_2}}.
  \label{eq:h3fam2}
\end{equation}
We prove the following theorem.
\begin{theorem}
  \label{thm:sup1}
  Hamiltonian system given by
  \begin{equation*}
    H=2p_1p_2 +\frac{q_2^d}{\sqrt{q_1}}
  \end{equation*}
  is super-integrable with algebraic additional first
  integral~\eqref{eq:h3fam2} iff $d=p$, or $d=-(2p-1)/2$, for a
  positive integer $p$.
\end{theorem}
\begin{proof}
  The first integral~\eqref{eq:h3fam2} is algebraic iff the
  hypergeometric function $\phantom{\vert}_2F_{1}$ appearing in its
  definition is algebraic.  This implies that the hypergeometric
  equation~\eqref{eq:hipcio} has an algebraic solution holomorphic at
  the origin.

  For equation~\eqref{eq:hipcio} the differences of exponents at
  singularities are
  \begin{equation}
    \lambda=1-c=\frac{1}{2},\quad \nu=c-a-b=d+1,\quad \mu=b-a=-d+\frac{1}{2}.
  \end{equation}
  This is why, by Lemma~\ref{lem:redrie}, see Appendix, the Gauss
  hypergeometric equation~\eqref{eq:hipcio} is reducible.  In fact,
  for an arbitrary value of $d$ we have
  \begin{equation}
    \lambda+\mu+\nu=2,\quad -\lambda+\mu+\nu=1,\quad
    \lambda-\mu+\nu=2d+1,\quad \lambda+\mu-\nu=-2d.
    \label{eq:combinations0}
  \end{equation}
  Thus,  always at least one of this expressions is odd integer, and this
  proves our claim.

  As equation~\eqref{eq:hipcio} is reducible its one solution is
  algebraic.  In fact $w_1=\sqrt{x}$ is a solution of this equation.
  As we remarked, if first integral~\eqref{eq:h3fam2} is algebraic,
  then equation~\eqref{eq:hipcio} has a solution $w_{2}$ which is
  algebraic and holomorphic at the origin.  Obviously $w_2$ is
  linearly independent with $w_1$. In effect all solutions of
  equation~\eqref{eq:hipcio} are algebraic.

  Now we apply Theorem~\ref{thm:kimura} in order to find values of $d$
  for which equation \eqref{eq:hipcio} has only algebraic
  solutions. The first assumption of this theorem is that all
  exponents at all singularities are rational. For equation
  \eqref{eq:hipcio} exponents at $x=0$, $x=1$, and $x=\infty$ are
  given by
  \[
  \left\{0,\frac{1}{2}\right\},\qquad \left\{0,d+1\right\},\qquad
  \left\{-\frac{1}{2},-d\right\}.
  \]
  Thus, if all solutions of equation \eqref{eq:hipcio} are algebraic,
  then $d$ must be rational.

  The next assumption of the theorem is that among four
  numbers~\eqref{eq:combinations0} exactly two or four are odd
  integers. This condition implies that $d=p/2$, with $p\in\Z$.

  Now we have to check the presence of logarithms in local solutions
  of \eqref{eq:hipcio} around particular singularities. At first we
  assume that $p$ is even i.e. $p=2q$ and $d=q\in\Z$. In this case the
  difference of exponents is integer only for singularity $x=1$.  For
  $q\in\N$, the difference of exponents is $m=q+1-0=q+1$. Hence,
  according to Appendix, the set $\langle m\rangle$ is $\langle
  m\rangle=\{1,2,\ldots,q+1\}$, see \eqref{eq:mym}.  Now we have to
  check if
  \begin{equation}
    e_{1,1}+e_{0,i}+e_{\infty,j}\not\in \langle m\rangle,\quad\text{for}\quad i,j\in\{1,2\},
    \label{eq:diflogi}
  \end{equation}
  with $e_{1,1}=q+1$. But we see that $q+1+0-q=1\in \langle
  m\rangle$. Thus for $d=p/2=q$ positive integer solutions of
  \eqref{eq:hipcio} contain no logarithms. For $q$ negative integer
  difference of exponents is $m=0-(q+1)=-q-1$, so $\langle
  m\rangle=\{1,2,\ldots,-q-1\}$.  Now we have to check all conditions
  \eqref{eq:diflogi} with $e_{1,1}=0$. We obtain successively
  $0+0-1/2=-1/2\not\in \langle m\rangle$, $0+0-q=-q\not\in \langle
  m\rangle$, $0+1/2-1/2=0\not\in \langle m\rangle$, and
  $0+1/2-q=-q+1/2\not\in \langle m\rangle$. This means that if $q$ is
  a negative integer, then logarithmic terms appear in solutions of
  \eqref{eq:hipcio}. Similar considerations for $p=2q-1$ show that  logarithmic terms
  appear in local solutions of \eqref{eq:hipcio} around $x=\infty$  iff
  $d=p/2=(2q-1)/2$, with positive integer $q$.
  
When hypergeometric function $\phantom{\vert}_2F_{1}$ is an
  algebraic function of its argument, then $H_3$ defined in
  \eqref{eq:h3fam2} is an algebraic function.

  Finally, one can check by direct substitution that for $d=0$
  hypergeometric function with the above parameters is equal to 1, and
  this ends our proof.
\end{proof}
\subsection{Additional integrals polynomial in momenta}

Let us notice that in general the first integral given
by~\eqref{eq:h3fam2} for values of $d$ specified in
Theorem~\ref{thm:sup1} is not polynomial in the momenta.  However, we
can show the following.
\begin{lemma}
  If $d=p\in\N$, then
  \begin{equation}
    \label{eq:h3p1}
    \widetilde H_3= \widetilde H_{3,p}:= H_3H_1^p,
  \end{equation}
  with $H_3$ given by~\eqref{eq:h3fam2} is a first integral of the
  system generated by Hamiltonian
  \begin{equation*}
    H=H_1= 2p_1p_2 + \frac{q_2^p}{\sqrt{q_1}}.
  \end{equation*}
  Moreover, $ \widetilde H_3 $ is a polynomial in momenta and has
  degree $2p+1$.
\end{lemma}
\begin{proof}
  Clearly $\widetilde H_3$ is a first integral of $H$.  In order to
  show that $\widetilde H_3$ is a polynomial in momenta we notice that
  \begin{equation*}
    \widetilde H_3= H_3H_1^p=
    \frac{1}{\sqrt{2q_1}p_1}\left[\sqrt{q_1}H_1^{p+1}-
I_2^pW_p\left( \frac{-2p_1^2q_1}{I_2}\right)\right],
  \end{equation*}
  where
  \begin{equation*}
    W_p(x):= \phantom{\vert}_2F_{1}\left(-\frac{1}{2},-p,
      \frac{1}{2}, x\right), \qquad W_p(0) = 1,
  \end{equation*}
  is a polynomial of degree $p$, see Appendix.  Thus,
  \begin{equation*}
    B:=\sqrt{q_1}H_1^{p+1}-I_2^pW_p\left( \frac{-2p_1^2q_1}{I_2}\right),
  \end{equation*}
  is a polynomial in momenta of degree $2p+2$.  Moreover we have
  \begin{equation*}
    \sqrt{q_1}H_1^{p+1}=\frac{q_2^{p(p+1)}}{q_1^{p/2}}+\cdots, \mtext{and}
    I_2^pW_p\left( \frac{-2p_1^2q_1}{I_2}\right)=
\frac{q_2^{p(p+1)}}{q_1^{p/2}}+\cdots,
  \end{equation*}
  where dots denote terms of degree at lest one in $p_1$.  This shows
  that $B$ is divisible by $p_1$. In effect $\widetilde H_3$ is
  polynomial in the momenta and has degree $2p+1$.
\end{proof}
\begin{lemma}
  \label{lem:h3p2}
  If $d= -(2p-1)/2$ for a certain $p\in\N$, then
  \begin{equation}
    \label{eq:h3p2}
    \widetilde H_3= \widetilde H_{3,p}:= H_3I_2^{p}/H_1,
  \end{equation}
  with $H_3$ and $I_2$ given by~\eqref{eq:h3fam2} and \eqref{eq:i2c2}
  is a first integral of system generated by Hamiltonian
  \begin{equation*}
    H=H_1= 2p_1p_2 + \frac{q_2^p}{\sqrt{q_1}}.
  \end{equation*}
  Moreover $ \widetilde H_3 $ is a polynomial in momenta and has
  degree $2p-1$.
\end{lemma}
\begin{proof}
  First we show that for the assumed values of $d$ we have
  \begin{equation*}
    \phantom{\vert}_2F_{1}\left(-\frac{1}{2},-d,
      \frac{1}{2}, x\right) = \phantom{\vert}_2F_{1}\left(-\frac{1}{2},\frac{1}{2}(2p-1),
      \frac{1}{2}, x\right) = (1-x)^{-(2p-3)/2}W_{p-1}(x), \qquad
  \end{equation*}
  where $W_{p-1}(x)$ is a polynomial of degree $p-1$, and
  $W_{p-1}(0)=1 $.

  In fact, for an arbitrary $p$, hypergeometric
  equation~\eqref{eq:hipcio} has in a neighbourhood of $x=1$ a local
  solution of the form $w= (1-x)^{-(2p-3)/2}u(x)$ where $u(x)$ is a
  function holomorphic at the origin and $u(0)=1$.  This is so,
  because $-(2p-3)/2$ is an exponent at $x=1$ for this equation.  In
  order to find $u(x)$ we make substitution $w= (1-x)^{-(2p-3)/2}u(x)$
  in equation~\eqref{eq:hipcio}.  The obtained equation is again a
  hypergeometric equation. Its solution holomorphic at the origin is
  \begin{equation}
    \label{eq:u}
    u(x)=  \phantom{\vert}_2F_{1}\left(1,1-p,
      \frac{1}{2}, x\right).
  \end{equation}
  As $p$ is a positive integer, $u(x)$ is a polynomial of degree $p-1$.

  Now, we have
  \begin{equation*}
    \label{eq:}
    \widetilde H_3:= H_3I_2^{p}/H_1= \frac{1}{\sqrt{2q_1}\,p_1}\left[
      \sqrt{q_1} I_2^p -
q_2^{-(2p-3)/2} I_2^{p-1}W_{p-1}\left( -\frac{p_1^2q_1}{I_{2}}\right)  \right].
  \end{equation*}
  As in the proof of the previous lemma it is easy to notice that
  expression in the square bracket is a polynomial in momenta which is
  divisible by $p_1$, and, as a result, $\widetilde H_3$ is a
  polynomial with respect to the momenta of degree $2p-1$.
\end{proof}
\subsection{Examples}
Here we consider two examples of super-integrable systems given by
Theorem~\ref{thm:sup1}.  Let us take $d=p=2$. In this case the
Hamiltonian function has the form
\[
H=H_1=2 p_1 p_2 + \frac{q_2^2}{\sqrt{q_1}}.
\]
Then, according Theorem~\ref{prop12}, first integral $I_2$ is
following
\[
I_2=2 p_1 (p_2 q_2-p_1 q_1)+\frac{q_2^3}{\sqrt{q_1}}.
\]
The second additional first integral
$\widetilde{H}_3=\widetilde{H}_{3,p}$ which guarantees
the super-integrability is given by \eqref{eq:h3p1}. For $p=2$, it has, up
to a multiplicative constant, the following form
\[
\widetilde{H}_3=q_1^{-1}\left[-16 p_1^3 q_1^{5/2} + 3 p_2 q_2^4 + 12
  p_1^2 p_2 q_1 (p_2^2 + 2 \sqrt{q_1}\, q_2) + 12 p_1 (p_2^2
  \sqrt{q_1}\, q_2^2 + q_1 q_2^3) \right].
\]
Moreover, there exists other first integral $H_4$ of the form
\[
H_4=-18 p_1 p_2^3+q_1^{-1/2}[48 p_1^2 q_1^2 - 72 p_1 p_2 q_1 q_2 - 9
p_2^2 q_2^2]-32 q_2^3.
\]
It is functionally independent with $H_1$ and $I_2$.  Integrals $H_1$,
$I_2$, $ \widetilde{H}_3$ and $H_4$ are algebraically dependent as we
have the following identity
\[
H_1^3 H_4 + \widetilde{H}_3^2 + 32 I_2^3=0.
\]
We note that $I_2^{-1/\alpha}$, $\widetilde{H}_3^{-2/(3\alpha)}$,
$H_4^{-2/\alpha}$ satisfy system \eqref{par-eq} with normalisation
$\kappa_2=1$.

Examples with other values of $p$ allow us to formulate the following.
\begin{conjecture}
  In the case 2 of Theorem~\ref{prop12}, if $d=p\in\N$, then there
  exists $\alpha_p\in\R$, such that first integral $H_4$ defined by the
  following relation
  \begin{equation*}
    H_1^{p+1} H_4 +  \widetilde{H}_{3,p}^2 + \alpha_{p} I_2^{p+1}=0,
  \end{equation*}
  is polynomial in the momenta of degree $2p$. Moreover, integrals
  $H_1$, $I_2$ and $H_4$ are functionally independent.
\end{conjecture}

Now, we consider the element of family with $d=-(2p - 1)/2$ for $p=3$. Then 
Hamilton is
\[
H_1=H=2 p_1 p_2 + \frac{1}{\sqrt{q_1} q_2^{5/2}}
\]
and first integral $I_2$ takes the form
\[
I_2=2 p_1 (p_2 q_2-p_1 q_1)+\frac{1}{\sqrt{q_1} q_2^{3/2}}.
\]
The explicit form of additional first integral $\widetilde
H_3=\widetilde H_{3,p}$ defined by \eqref{eq:h3p2} is, up to a
multiplicative constant, equal to
\[
\widetilde{H}_3=4 p_1^2 (p_1 q_1 - p_2 q_2)^3-\frac{4 p_1 (5 p_1^2
  q_1^2 - 12 p_1 p_2 q_1 q_2 + 3 p_2^2 q_2^2)}{ 3 \sqrt{q_1}
  q_2^{3/2}}+\frac{5 p_1 q_1 - p_2 q_2}{q_1 q_2^3}.
\]
One can check that
\[
H_4=18 p_1 (p_1 q_1 - p_2 q_2)^3+\frac{-33 p_1^2 q_1^2 + 90 p_1 p_2
  q_1 q_2 - 9 p_2^2 q_2^2}{\sqrt{q_1} q_2^{3/2}}+\frac{32}{q_2^3},
\]
is also a first integral of degree $2p-2=4$ with respect to the momenta,
and that the following relation is satisfied
\[
32H_1^3-9\widetilde{H}_3^2-I_2^3H_4=0.
\]
Moreover, $I_2^{-1/\alpha}$, $\widetilde{H}_3^{-1/(2\alpha)}$,
$H_4^{-1/\alpha}$ satisfy system \eqref{par-eq} with normalisation
$\kappa_2=1$.

As in the previous case, additional examples justify the following
conjecture.
\begin{conjecture}
  In the case 2 of Theorem~\ref{prop12}, if $d=-(2p - 1)/2$ with
  $p\in\N$, then there exists $\alpha_p\in\R$, such that first integral
  $H_4$ defined by the following relation
  \begin{equation*}
    H_1^{p}  +  \alpha_{p} \widetilde{H}_{3,p}^2   -I_2^{p}H_4=0,
  \end{equation*}
  is polynomial in the momenta of degree $2p-2$. Moreover, integrals
  $H_1$, $I_2$ and $H_4$ are functionally independent.
\end{conjecture}

\section{Case 3. Super-integrability without additive separation of
  variables}
\label{sec:c3}
In case 3 of Theorem~\ref{prop12} we have
\begin{equation}\label{h-case2}
  \beta=-\frac{2\alpha}{d+1}, \qquad
\rho=\alpha\left(p_1q_1-\dfrac{2p_2q_2}{d+1}\right), \qquad
  H_1=2p_1p_2+q_1{q_2^d},
\end{equation}
and one solution of the equations \eqref{par-eq} is
equal to
\begin{equation}\label{int-case4}
  H_2=I_2^{-1/2\alpha}, \qquad I_2=p_1^2+\kappa {q_2^{d+1}},\qquad
  \kappa=\dfrac{1}{d+1}.
\end{equation}
As in the previous section our aim is to distinguish  cases when the
system is super-integrable.
\subsection{Algebraic super-integrability}
In order to find the second additional first integral $H_3$ we proceed as
follows.  From the integral $I_2$ \eqref{int-case4} we directly obtain
the first quadrature
\begin{equation}
  \label{eq:bet1}
  \beta_1+t=-\int^{q_2} \dfrac{ \mathrm{d}x}{2\sqrt{I_2-\kappa\,x^{d+1}}}.
\end{equation}
After calculation of this quadrature we have to substitute solution
$q_2=q_2(t,\beta_1)$ into the Hamiltonian
\[
I_1=\dfrac{1}{2}
\dfrac{\mathrm{d}q_1}{\mathrm{d}t}\,\dfrac{\mathrm{d}q_2}{\mathrm{d}t}+
q_1q_2^d,
\]
and solve the obtained quadrature with respect to $q_1$.  However to
find $q_2=q_2(t)$ we have to invert explicitly
quadrature~\eqref{eq:bet1}, and it is unclear how to do this. By this  reason 
 we proceed in a different way. Using function $\rho$ and
algebraic relations \eqref{rho-w} we can get the second angle variable
\begin{equation}\label{case53}
  \omega_2=\dfrac{\rho}{2\alpha
    I_2}+\dfrac{d-1}{d+1}\,\dfrac{I_1}{2I_2}\,\int^{q_2} \dfrac{
    \mathrm{d}x}{2\sqrt{I_2-\kappa\,x^{d+1}}}\,.
\end{equation}
In generic case $\omega_2$ is the multi-valued function on the whole
phase space \cite{ts07f,ts08a,ts08b}.

We want to distinguish values of $d$ for that system in case 3 from
Theorem~\ref{prop12} governed by Hamiltonian $H_1$ given by
\eqref{h-case2} is super-integrable. Since $H_1=I_1$, thus $\omega_2$
is a first integral.  So, the following function
\[
H_3=4\alpha(d+1)I_2\omega_2=2(d+1)\rho+(d-1)\alpha I_1 \,\int^{q_2}
\dfrac{ \mathrm{d}x}{\sqrt{I_2-\kappa\,x^{d+1}}}\,.
\label{eq:H3}
\]
is also a first integral which we are going to analyse in details.

In the integral presented in $H_3$ we make the following substitution
\[
y=\dfrac{\kappa}{I_2}x^{d + 1},
\]
which gives
\[
\int \dfrac{
  \mathrm{d}x}{\sqrt{I_2-\kappa\,x^{d+1}}}=\dfrac{1}{(d+1)\sqrt{I_2}}\left(\dfrac{I_2}{\kappa}
\right)^{\frac{1}{d+1}} \int \dfrac{ y^{-\frac{d}{d+1}}
  \mathrm{d}y}{\sqrt{1-y}}.
\]
Let us denote
\begin{equation}
  \int \dfrac{ y^{-\frac{d}{d+1}}
    \mathrm{d}y}{\sqrt{1-y}}=(d+1)y^{\frac{1}{d+1}}v(y),
\end{equation}
where $v(y)$ is a certain function of $y$.  Differentiating both sides
of this equation with respect to $y$, after obvious simplifications we
obtain the following linear non-homogeneous equation for the unknown
function $v(y)$
\begin{equation}
  \dfrac{1}{\sqrt{1-y}}=v+(d+1)yv'.
  \label{eq:nainta}
\end{equation}
Multiplying it by $\sqrt{1-y}$ and differentiating once more we obtain
\begin{equation}
  \label{eq:h3}
  y(1-y)v''+\left(\dfrac{d+2}{d+1}-\dfrac{3 d+5}{2 d+2}y
  \right)v'-\dfrac{1}{2 d+2}v=0.
\end{equation}
This is the Gauss hypergeometric equation with parameters
\[
a =\frac{1}{2},\quad b=\frac{1}{d + 1},\quad c=1+\frac{1}{d+1}.
\]
General solution of~\eqref{eq:h3} has the form
\begin{equation}
  \label{eq:solh3}
  v(y)=\beta_2y^{-1/(d+1)} + \beta_3 \phantom{\vert}_2F_{1}\left(\frac{1}{2},\frac{1}{d+1},
    1+\frac{1}{d+1},y\right).
\end{equation}
Substituting it into equation~\eqref{eq:nainta} we find out that it is
its solution iff $\beta_3=1$. Without loss of the generality we can
assume that $\beta_2=0$, and then in variable $x$ considered integral
is
\begin{equation*}
  \int \dfrac{
    \mathrm{d}x}{\sqrt{I_2-\kappa\,x^{d+1}}}=\frac{x}{\sqrt{I_2}}\,
  \phantom{\vert}_2F_{1}\left(\frac{1}{2},\frac{1}{d+1},
    1+\frac{1}{d+1},\frac{\kappa x^{d+1}}{I_2}\right).
\end{equation*}
Thus our first integral takes the final form
\begin{equation}
  \label{eq:h3t}
  H_3=2(d+1)\rho+(d-1)\alpha \dfrac{I_1q_2}{\sqrt{I_2}}\, \phantom{\vert}_2F_{1}\left(\frac{1}{2},\frac{1}{d+1},
    1+\frac{1}{d+1},\frac{\kappa q_2^{d+1}}{I_2}\right).
\end{equation}
This first integral is functionally independent with $I_1$ and $I_2$.
In general hypergeometric functions are not single-valued with sinularities of  non-algebraic character.    Thus, in
order to obtain first integrals with ``good'' behaviour,   we will
restrict to cases when function $\phantom{\vert}_2F_{1}$ is
algebraic. In other words, we are looking for such values of $d$ that
equation \eqref{eq:h3} has a solution that is algebraic.  At first let
us note that by Lemma~\ref{lem:redrie} equation \eqref{eq:h3} is
reducible for any value of $d$. Indeed, the differences of exponents
at singularities of equation \eqref{eq:h3} are
\begin{equation}
  \lambda=1-c=-\frac{1}{d+1},\quad \nu=c-a-b=\frac{1}{2},\quad
  \mu=b-a=\frac{1}{d+1}-\frac{1}{2},
\end{equation}
and among numbers
\begin{equation}
  \lambda+\mu+\nu=0,\quad -\lambda+\mu+\nu=\frac{2}{d + 1},\quad
  \lambda-\mu+\nu=1-\frac{2}{d + 1},\quad \lambda+\mu-\nu=-1,
  \label{eq:combinations}
\end{equation}
we have one even and one odd number. This means that we can use
Theorem~\ref{thm:kimura} in order to find values of $d$ for which
equation~\eqref{eq:h3} has only algebraic solutions.
\begin{theorem}
  \label{thm:3}
  Hamiltonian system given by
  \begin{equation*}
    H=2p_1p_2 +q_1q_2^{d},
  \end{equation*}
  is super-integrable with algebraic additional first
  integral~\eqref{eq:h3t} iff $d$ takes the form
  \begin{equation}
    d=\frac{1 - p}{p}, \quad \text{or}\quad d=\frac{1+2p}{1-2p},
    \label{eq:condony}
  \end{equation}
  for $p\in\N$.
\end{theorem}
\begin{proof}
  We apply directly Theorem~\ref{thm:kimura}. Exponents at
  singularities $x=0$, $x=1$ and $x=\infty$ for equation~\eqref{eq:h3}
  are
  \[
  \left\{0,-\frac{1}{d+1}\right\},\qquad
  \left\{0,\frac{1}{2}\right\},\qquad
  \left\{\frac{1}{2},\frac{1}{d+1}\right\},
  \]
  respectively. The first condition of Theorem~\ref{thm:kimura}
  implies that $d$ has to be rational.  The second condition can be
  fulfilled only when two among number \eqref{eq:combinations} are
  odd, and this implies that
  \begin{equation}
    \label{eq:condo}
    \frac{2}{d + 1}=r, \mtext{so} d=\frac{2 - r}{r},
  \end{equation}
  for a certain $r\in\Z$.

  Now we have to check the  presence of logarithms in local solutions of
  \eqref{eq:h3} around particular singularities. In this aim we
  apply Lemma~\ref{lem:iwa3}. For $d$ given by \eqref{eq:condo}
  exponents are following
  \[
  \left\{0,-\frac{r}{2}\right\},\qquad
  \left\{0,\frac{1}{2}\right\},\qquad
  \left\{\frac{1}{2},\frac{r}{2}\right\}.
  \]
  We see that differences of exponents can be integer positive only
  for $x=0$ or $x=\infty$ depending on the parity of $r$ as well as the
  positive or negative sign of $r$. An analysis similar to that done
  in the proof of Theorem~\ref{thm:sup1} shows that  logarithmic terms
  appear in local solutions iff either $r$ is positive and odd, or $r$
  is negative and even.

  If all solutions of the considered equation are algebraic then, by
  Theorem~\ref{thm:kimura} none of local solution has a logarithmic
  term. Thus, by the above considerations, either $r$ is a
  positive and even integer, or it is a negative and odd integer.
  Substituting $r= 2p$ as well as $r= -(2p-1)$ with $p\in\N$ into \eqref{eq:condo}
  we obtain \eqref{eq:condony}, and this ends the proof.
\end{proof}
\subsection{Additional first integrals polynomial in the momenta}
If assumptions of Theorem~\ref{thm:3} are fulfilled, then first integral
given by~\eqref{eq:h3t} is algebraic but it is not polynomial with
respect to the momenta. However we can find such integrals using the
following.
\begin{lemma}
\label{lem:th3}
  If
  \begin{equation*}
    d=\frac{1-p}{p} \mtext{for a certain} p\in \N,
  \end{equation*}
  then $H_3$ given by~\eqref{eq:h3t} can be  uniquely written in the
  following form
  \begin{equation}
    \label{eq:th3}
    H_3:=\widetilde H_{3,p}+\sqrt{I_2}J_{3,p},
  \end{equation}
  where $ \widetilde H_{3,p}$ and $J_{3,p}$ are first integrals polynomial in the
  momenta of the system with Hamiltonian  $H=H_1$ defined
  by~\eqref{h-case2}. Moreover,  $ \widetilde H_{3,p}$   has degree $2p-1$ with
  respect to the moment and is functionally independent together with $H$
  and $I_2$.
\end{lemma}
\begin{proof}
  For the specified form of $d$, we have
  \begin{equation}
    \label{eq:f}
    F(x):= \phantom{\vert}_2F_{1}\left(\frac{1}{2},\frac{1}{d+1},
      1+\frac{1}{d+1},x \right)= \phantom{\vert}_2F_{1}\left(\frac{1}{2},p,
      1+p,x \right).
  \end{equation}
  We show that
  \begin{equation}
    \label{eq:fffa}
    F(x)=x^{-p}\left(C + \sqrt{1-x}W_{p-1}(x)\right),
  \end{equation}
  where $W_{p-1}(x)= \phantom{\vert}_2F_{1}\left(\frac{1}{2},1-p,
    1-p,x \right)$ is a polynomial of degree $p-1$, and
  \begin{equation*}
    C= \!\phantom{\vert}_2F_{1}\left(\frac{1}{2},p,
      1+p,1\right)=\sqrt{\pi}\frac{\Gamma(1+p)}{\Gamma \left( \frac{1}{2}+p \right)}=\frac{2^pp!}{(2p-1)!!}.
  \end{equation*}
  In fact, $v:=F(x)$ is a solution of hypergeometric
  equation~\eqref{eq:h3} holomorphic at the origin.  Making
  substitution~\eqref{eq:fffa} we find that $W_{p-1}(x)$ coincides with
  $\phantom{\vert}_2F_{1}\left(\frac{1}{2},1-p, 1-p,x \right) $.  As
  $p$ is a positive integer $W_{p-1}(x)$ is a polynomial of degree
  $p-1$.

  In the considered case we have
  \begin{equation*}
    x:= \frac{q_2^{d+1}}{(d+1)I_2}=\frac{pq_2^{1/p}}{I_2},
    \mtext{and}1-x=\frac{p_1^2}{I_2}.
  \end{equation*}
  Thus we can rewrite $H_3$ given by~\eqref{eq:h3t} in the following
  form
  \begin{equation}
    \label{eq:h3re}
    H_3 = \frac{2}{p}\rho +\alpha
    \frac{1-2p}{p^{p+1}}I_1\sqrt{I_2}I_2^{p-1}\left(
      C+\frac{p_1}{\sqrt{I_2}}W_{p-1}\left( \frac{pq_2^{1/p}}{I_2} \right) \right),
  \end{equation}
 i.e., it   has  the form~\eqref{eq:th3} with
  \begin{equation}
    \label{eq:31}
   \widetilde H_{3,p}:= \alpha  \frac{2}{p}(p_1q_1-2pp_2q_2) +\alpha
    \frac{1-2p}{p^{p+1}}p_1I_1 I_2^{p-1}W_{p-1}\left( \frac{pq_2^{1/p}}{I_2} \right) ,
  \end{equation}
  and
  \begin{equation}
    \label{eq:32}
    J_{3,p}:= \alpha C
    \frac{1-2p}{p^{p+1}}I_1 I_2^{p-1}.
  \end{equation}
  Clearly, these are polynomials with respect to momenta, and
  $\widetilde H_{3,p}$
  has degree $2p-1$.
\end{proof}

\begin{lemma}
  If
  \begin{equation}
    \label{eq:de1a}
    d=\frac{1+2p}{1-2p} \mtext{for a certain} p\in \N,
  \end{equation}
  then
  \begin{equation}
    \label{eq:th3a}
    \widetilde H_{3,p}:=H_3I_2^p,
  \end{equation}
  where $H_3$ given by~\eqref{eq:h3t} is a first integral polynomial
  in the momenta of degree $2p+1$.
\end{lemma}
\begin{proof}
  We show that if $d$ is given by~\eqref{eq:de1a}, then
  \begin{equation}
    \label{eq:ff}
    \phantom{\vert}_2F_{1}\left(\frac{1}{2},\frac{1}{d+1},
      1+\frac{1}{d+1},y\right)=
    \phantom{\vert}_2F_{1}\left(\frac{1}{2},\frac{1}{2}-p,
      \frac{3}{2}-p,y\right)=\sqrt{1-y}R_{p-1}(y),
  \end{equation}
  where $R_{p-1}(y)$ is a polynomial of degree $p-1$.

  Recall that hypergeometric function given above is a solution of
  equation~\eqref{eq:h3}. If we make substitution
  \begin{equation}
    \label{eq:s3}
    v(y)=\sqrt{1-y}u(y)
  \end{equation}
  in equation~\eqref{eq:h3}, then we obtain again the Gauss hypergeometric
  equation
  \begin{equation}
    \label{eq:hypu}
    y(1-y)u'' + \left( \frac{3}{2} -p +(p-3)y\right)u'+(p-1)u=0,
  \end{equation}
  with  the parameters 
  \begin{equation}
    \label{eq:uabc}
    a=1-p, \qquad b=1, \qquad c = \frac{1}{2}(3-2p).
  \end{equation}
  For $p\in \N$, parameter $a$ is non-positive integer, thus function
  $u(y):=\phantom{\vert}_2F_{1}\left(a,b,c,y\right)$ is a polynomial
  of degree $p-1$.  Hence, we have
  \begin{multline}
    \label{eq:ffu}
    v(y)= \phantom{\vert}_2F_{1}\left(\frac{1}{2},\frac{1}{2}-p,
      \frac{3}{2}-p,y\right)= \sqrt{1-y}u(y)=\\
    \sqrt{1-y}\phantom{\vert}_2F_{1}\left( 1-p,1
      ,\frac{1}{2}(3-2p),y\right).
  \end{multline}
  It follows that $R_{p-1}(y)$ is equal to
  $\phantom{\vert}_2F_{1}\left( 1-p,1 ,\frac{1}{2}(3-2p),y\right)$,
  and so it is a polynomial of degree $p-1$, as we claimed.

  Now we have to show that $\widetilde H_{3,p}$ given by~\eqref{eq:th3a} is
  a polynomial with respect to the momenta.  But we can write
  \begin{multline}
    \label{eq:th3t}
    \widetilde H_{3,p}=\alpha I_2^p \left( 2(d+1)q_1p_1-q_2p_2
    \right)+\alpha(d-1)I_2^p\frac{I_1q_2}{\sqrt{I_2}}\sqrt{1-\frac{\kappa
        q_2^{d+1}}{I_2}}R_{p-1}\left( \frac{\kappa q_2^{d+1}}{I_2}
    \right)=\\
    =\alpha I_2^p \left( 2(d+1)q_1p_1-q_2p_2
    \right)+\alpha(d-1)I_1q_2\sqrt{I_2}\sqrt{1-\frac{\kappa
        q_2^{d+1}}{I_2}}S_{p-1},
  \end{multline}
  where
  \begin{equation}
    \label{eq:sp}
    S_{p-1}
    =I_2^{p-1}R_{p-1}\left( \frac{\kappa q_2^{d+1}}{I_2},
    \right),
  \end{equation}
  is a polynomial of degree not higher than $2(p-1)$ with respect to
  the momenta. Now, it is enough to notice that according
  to~\eqref{eq:bet1} and~\eqref{int-case4}, we have
  \begin{equation*}
    \sqrt{I_2}\sqrt{1-\frac{\kappa
        q_2^{d+1}}{I_2}}= -p_1,
  \end{equation*}
  in order to see that both terms in~\eqref{eq:th3t} are polynomial
  with respect to the momenta. Moreover, it is easy to see that the
  degree of $\widetilde H_{3,p}$ with respect to the momenta is $2p+1$.
\end{proof}

\subsection{Examples}

As example we take the family with $d=(1-p)/p$, and we choose $p=5$. In
this case we have
\[
H_1=2 p_1 p_2 + \frac{q_1}{q_2^{4/5}},\qquad I_2=p_1^2 + 5 q_2^{1/5}.
\]
Up to a multiplicative constant, the first integral $\widetilde{H}_3:=
\widetilde{H}_{3,p}$ defined in Lemma~\ref{lem:th3} by
equation~\eqref{eq:th3} has the following form
\[
\begin{split}
  \widetilde{H}_3=&q_2^{-4/5}[128 p_1^8 p_2 q_2^{4/5}+64 p_1^6 (p_1
  q_1 + 35 p_2 q_2)+
  560 p_1^4 q_2^{1/5} (2 p_1 q_1 + 25 p_2 q_2)\\
  &+7000 p_1^2 q_2^{2/5} (p_1 q_1 + 5 p_2 q_2)+4375 q_2^{3/5} (4 p_1
  q_1 + 5 p_2 q_2)].
\end{split}
\]
There exists also other first integral of degree $2p-2=8$ in the
momenta
\[
\begin{split}
  H_4= & 896p_1^6p_2(p_1q_1-p_2q_2)+224q_2^{-4/5}p_1^4(2p_1^2q_1^2+66p_1p_2q_1q_2-65p_2^2q_2^2)\\
  &+
  16q_2^{-3/5}p_1^2(476p_1^2q_1^2+5215p_1p_2q_1q_2-5025p_2^2q_2^2)\\
  &+5q_2^{-2/5}(9072p_1^2q_1^2 +32920p_1p_2q_1q_2
  -30625p_2^2q_2^2)+102400q_1^2q_2^{-1/5},
\end{split}
\]
and the following relation holds
\[
3125 H_4 +\widetilde{H}_3^2 - 4096 H_1^2 I_2^7=0.
\]
Moreover, $I_2^{-1/(2\alpha)}$, $\widetilde{H}_3^{1/(2\alpha)}$,
$H_4^{1/(4\alpha)}$ satisfy system \eqref{par-eq} with normalisation
$\kappa_2=1$.  Other examples justify the following.
\begin{conjecture}
  In the case 3 of Theorem~\ref{prop12}, if $d=(1-p)/p$ for a
  certain $p\in\N$, then there exists $\alpha_p\in\R$, such that first
  integral $H_4$ given by
  \begin{equation*}
    H_4 =  \widetilde{H}_{3,p}^2 + \alpha_{p} I_2^{2p-3}H_1^{2}
  \end{equation*}
  is polynomial in the momenta of degree $2p-2$. Moreover, integrals
  $H_1$, $I_2$ and $H_4$ are functionally independent.
\end{conjecture}

Now, let $d=(1+2p)/(1-2p)$, and we consider the case with $p=2$, i.e.,
$d=-5/3$. In this case we have
\begin{equation}
\label{eq:h1i2}
  H_1=2p_1 p_2 + \dfrac{q_1}{q_2^{5/3}},\mtext{and}
  I_2=p_1^2-\dfrac{3}{2q_2^{2/3}}.
\end{equation}
The additional first integral
$\widetilde{H}_3:=\widetilde{H}_{3,p}$ defined by \eqref{eq:th3a} of
degree  $2p+1$ with respect to the momenta  is, up to a multiplicative constant, of
the following form
\begin{equation}
\label{eq:h3d53}
\widetilde{H}_3=4 p_1^4 (p_2 q_2 - p_1 q_1)+\frac{4 p_1^2 (5 p_1 q_1
  - 9 p_2 q_2)}{q_2^{2/3}}- \frac{9 (5 p_1 q_1 + 3 p_2
  q_2)}{q_2^{4/3}}.
\end{equation}
We have also first integral $H_4$  of degree $2p$ with respect to the
momenta given by
\begin{equation}
  \label{eq:hh3}
  H_4=p_1^2(p_1q_1-p_2q_2)^2-\dfrac{(11p_1q_1-27p_2q_2)(p_1q_1-p_2q_2)}{2q_2^{2/3}}+
\dfrac{16q_1^2}{q_2^{4/3}},
\end{equation}
and the following identity holds
\[
864 H_1^2 -\widetilde{H}_3^2 + 16 H_4 I_2^3=0.
\]
Moreover, $I_2^{-1/(2\alpha)}$, $\widetilde{H}_3^{-1/(4\alpha)}$,
$H_4^{-1/(2\alpha)}$ satisfy system \eqref{par-eq} with normalisation
$\kappa_2=1$.

Other examples with different values of $p$ allow us to formulate the following.
\begin{conjecture}
  In the case 3 of Theorem~\ref{prop12}, if $d=(1-2p)/(1+2p)$ for a
  certain $p\in\N$, then there exists $\alpha_p\in\R$, such that first
  integral $H_4$ defined by
  \begin{equation*}
    H_4I_2^{p+1} + \widetilde{H}_{3,p}^2 + \alpha_{p} H_1^{2}=0
  \end{equation*}
  is polynomial in the momenta of degree $2p$. Moreover, integrals
  $H_1$, $I_2$ and $H_4$ are functionally independent.
\end{conjecture}

Let us determine the complete algebra generated by first
integrals~\eqref{eq:h1i2}, \eqref{eq:h3d53}, \eqref{eq:hh3}  and
$\rho$ which in  the case $d = -5/3$ has the following form
\begin{equation*}
\rho =-\frac{1}{2}(q_1p_1+3q_2p_2).
\end{equation*}
It is easy to check that first integral given by
\[
\widetilde H_3=\{I_2,H_4\}\,,
\]
satisfies equations \eqref{par-eq} with $\kappa_2=2$. The complete
algebra of integrals and $\rho$ is given by
\[
\begin{array}{llll}
  \{I_1,\rho\}=2I_1, \quad \{I_2,\rho\}=I_2,
  &\{H_4,\rho\}=H_4\,,\quad &\{\widetilde H_3,\rho\}=2\widetilde H_3
  ,\qquad \{I_2,H_4\}=\widetilde H_3,
  \\
  \{I_1,I_2\}=\{I_1,\widetilde H_3\}=\{I_1,H_4\}=0,  &\{I_2,\widetilde H_3\}=8 I_2^3,\quad
  &\{\widetilde H_3,H_4\}=24 I_2^2H_4.
\end{array}
\]
Using this algebra we easily obtain the desired angle variable
\[
\omega_2=-\dfrac{1}{8}\dfrac{F\left(\dfrac{I_2}{\sqrt{I_1}}\right)I_1+
4\sqrt{H_4I_2^3+54
    I_1^2}}{I_2^3},
\]
up to canonical transformation $\omega_2\to \omega_2+f(I_2)$. Here $F$
is arbitrary function of the constant of motion, chosen in such a way
that relation $\{\omega_1,\omega_2\}=0$ is satisfied.

\section{Case 4}
\label{sec:c4}
In case 4 of Theorem~\ref{prop12} we have
\begin{equation}\label{h-case1}
 H_1=2p_1p_2+q_1^{-d-2}q_2^d.
\end{equation}
The system generated by this  Hamiltonian  is integrable for an
arbitrary value of parameter $d$ with the following first
integral
\begin{equation}
  I_2=(p_1q_1-p_2q_2)^2-2\frac{q_2^{d+1}}{q_1^{d+1}}.
  \label{eq:hah2}
\end{equation}
This system admits separation of variables and we show that for
rational values of $d$ it is super-integrable.

\subsection{Super-integrability and additive separation of variables}

Let us introduce the following canonical coordinates $(r,\varphi)$ and
momenta $(p_r,p_{\varphi})$ defined by
\begin{equation}
  q_1=r(\cosh\varphi+ \sinh\varphi)=re^{ \varphi},\qquad
  q_2=r(\cosh\varphi- \sinh\varphi)=re^{-\varphi},
\end{equation}
and
\begin{equation}
  p_1=\frac{e^{-\varphi}}{2}\left(p_r+\frac{p_{\varphi}}{r}\right),\qquad
  p_2=\frac{e^{\varphi}}{2}\left(p_r-\frac{p_{\varphi}}{r}\right).
\end{equation}
In new variables Hamiltonian $H_1$ and first integral $J_2= I_2/2$
take the forms
\begin{equation}
  H_1=\frac{p_r^2}{2}-\frac{p_{\varphi}^2}{2r^2}+\frac{e^{-2(d+1)\varphi}}{r^2},
  \quad J_2=\frac{1}{2}p_{\varphi}^2-e^{-2(d+1)\varphi}.
\end{equation}
New Hamilton's equations read
\begin{equation}
  \begin{split}
    &\dot r=p_r,\qquad\qquad
 \dot p_r=\frac{-p_{\varphi}^2+2e^{-2(d+1)\varphi}}{r^3},\\
    &\dot \varphi=-\frac{p_{\varphi}}{r^2},\qquad\,\, \dot
    p_{\varphi}= \frac{2 (d+1)e^{-2(d+1)\varphi}}{r^2}.
  \end{split}
\end{equation}
Let us note that
\begin{equation}
  H_1=\frac{1}{2}p_r^2-\frac{1}{r^2}J_2.
\end{equation}
In order to perform the explicit integration we introduce as in
\cite{Borisov:09::} a new independent variable $\tau$ such that
$\mathrm{d}\tau/\mathrm{d}t=1/r^2$. Then we find that
\[
p_r=\dfrac{r'}{r^2},\qquad p_{\varphi}=-\varphi',
\]
where prime denotes the differentiation with respect to $\tau$.  In
effect we have
\[
H_1=\dfrac{r'^2}{2r^4}-\dfrac{1}{r^2}J_2,\mtext{and}
J_2=\dfrac{\varphi'^2}{2}-e^{-2(d+1)\varphi},
\]
i.e., we effectively separated variables and we can make two
independent quadratures
\begin{equation}
  \int \dfrac{\mathrm{d}r}{\sqrt{2(H_1r^4+J_2r^2)}}=\tau +C_1,\qquad
  \int
  \dfrac{\mathrm{d}\varphi}{\sqrt{2\left(J_2+e^{-2 (d+1)\varphi}\right)}}=
  \tau +C_2.
  \label{eq:ccalki0}
\end{equation}
The explicit forms of these elementary integrals are following
\begin{equation}
  \dfrac{1}{\sqrt{2J_2}}\ln\dfrac{r}{2(\sqrt{J_2}+\sqrt{J_2+H_1r^2}\,)}=\tau+C_1,
  \label{eq:sum10}
\end{equation}
and
\begin{equation*}
  \dfrac{1}{(d+1)\sqrt{2J_2}}\operatorname{arcsinh}\left[\sqrt{J_2}e^{ (d+1)\varphi}\right]=\tau+C_2.
\end{equation*}
Using well-known formula
\[
\operatorname{arcsinh}z=\ln [z+\sqrt{z^2+1}],
\]
one can rewrite the last integral as
\begin{equation}
  \dfrac{1}{(d+1)\sqrt{2J_2}}\ln\left[e^{ (d+1)\varphi}\sqrt{J_2}+\sqrt{e^{ 2(d+1)\varphi}J_2+1}
  \right]=\tau+C_2.
  \label{eq:sum20a}
\end{equation}
From \eqref{eq:sum10} and \eqref{eq:sum20a} we deduce that
\begin{equation}\begin{split}
    & I=(d+1)\sqrt{2J_2}(C_2-C_1)\\
    &=\ln\left[e^{ (d+1)\varphi}\sqrt{J_2}+\sqrt{e^{
          2(d+1)\varphi}J_2+1}
    \right]-(d+1)\ln\dfrac{r}{2(\sqrt{J_2}+\sqrt{J_2+H_1r^2}\,)}\\
    & = \ln\left[\frac{2^{d+1}}{r^{d+1}}\left(e^{
          (d+1)\varphi}\sqrt{J_2}+\sqrt{e^{
            2(d+1)\varphi}J_2+1}\right) \left(
        \sqrt{J_2}+\sqrt{J_2+H_1r^2}\right)^{d+1} \right],
  \end{split}
\end{equation}
is a first integral of the system.  We use relations
\[
J_2+H_1r^2=\frac{1}{2}r^2p_r^2,\qquad e^{
  2(d+1)\varphi}J_2+1=\dfrac{1}{2}e^{ 2(d+1)\varphi}p_{\varphi}^2,
\]
to introduce explicitly the momenta and then $I$ takes the form
\[
I=\ln\left[ \frac{2^{d/2}}{r^{d+1}}e^{
    (d+1)\varphi}(\sqrt{2J_2}-p_{\varphi})(\sqrt{2J_2}+rp_r)^{d+1}
\right],
\]
We can take as the first integral its exponent, more precisely
\begin{equation}
  H_3=2^{-d/2}\exp(I)=\frac{1}{r^{d+1}}e^{ (d+1)\varphi}(\sqrt{I_2}-p_{\varphi})(\sqrt{I_2}+rp_r)^{d+1}.
\end{equation}
When we come back to original variables, then this first integral
takes the form
\[
H_3=\frac{1}{q_2^{d+1}}(p_2q_2-p_1q_1+\sqrt{I_2}) \left(p_1 q_1 + p_2
  q_2+\sqrt{I_2}\right)^{d+1},
\]
It is easy to notice that
\[
H_4=\frac{1}{q_2^{d+1}}(p_2q_2-p_1q_1-\sqrt{I_2}) \left(p_1 q_1 + p_2
  q_2-\sqrt{I_2}\right)^{d+1}.
\]
is also a first integral.
For rational $d$ these functions are algebraic. Thus we can
recapitulate these considerations by the following theorem.
\begin{theorem}
  Hamiltonian system given by \eqref{h-case1} is super-integrable with
  algebraic additional first integral iff $d$ is a rational number.
\end{theorem}
Using $H_3$ and $H_4$ we can get first integrals which are polynomial
with respect to the momenta.  For
example, if $d$ is positive integer, then we take the following
first integrals
\[
\widetilde H_3=\frac{1}{\sqrt{I_2}}(H_3-H_4),\qquad \widetilde
H_4=H_3+H_4.
\]
In the sum $ H_3+H_4$  terms containing odd powers of
$\sqrt{I_2}$ disappear,  and, as a result, $\widetilde H_4$ is a
polynomial with respect to the momenta of degree
$d+2$.  The difference  $H_3-H_4$ contains only  odd powers of
$\sqrt{I_2}$ , thus it is divisible by $\sqrt{I_2}$,  and
$\widetilde H_3$ is polynomial in the momenta of degree $d+1$.

If  $d$ is a positive rational number of the form $d=d_1/d_2$, then we
put
\[
\begin{split}
  &F_3=H_3^{d_2}=\frac{1}{q_2^{d_1+d_2}}(p_2q_2-p_1q_1+\sqrt{I_2})^{d_2}
  \left(p_1 q_1 +
    p_2 q_2+\sqrt{I_2}\right)^{d_1+d_2},\\
  &F_4=H_4^{d_2}=\frac{1}{q_2^{d_1+d_2}}(p_2q_2-p_1q_1-\sqrt{I_2})^{d_2}
  \left(p_1 q_1 + p_2 q_2-\sqrt{I_2}\right)^{d_1+d_2},
\end{split}
\]
and then
\[
\widetilde H_3=\frac{1}{\sqrt{I_2}}(F_3-F_4),\qquad \widetilde
H_4=F_3+F_4,
\]
are  polynomial first integrals of degrees $d_1+2d_2-1$ and $d_1+2d_2$, respectively.

In the case when $d\geq-1$ is rational negative, i.e. $d=-d_1/d_2$,
and $d_2\geq d_1$ where
$d_1,d_2\in\N$ we take
\[
\widetilde H_3=\frac{1}{\sqrt{I_2}}(H_3^{d_2}-H_4^{d_2}) ,\qquad
\widetilde H_4=H_3^{d_2}+H_4^{d_2},
\]
as  polynomial in the momenta first integrals  of degrees
$2d_2-d_1-1$ and $2d_2-d_1$, respectively.

It is not clear if  there exist a polynomial  in momenta first
integral functionally independent with $H_1$ and $I_2$  for rational $d$ smaller than $-1$.

Let us note that potentials from case 4 written in Cartesian
coordinates as
\[
V=\frac{(x+\rmi y)^{k-1}}{(x-\rmi y)^{k+1}}
\]
were considered in the recent paper \cite{Kalnins:09}. In this paper
authors proved the super-integrability of such potentials for $k$
rational for classical and quantum systems. But the explicit form of
the additional first integral was not announced.

\subsection{Examples}
We showed that for rational $d$ the system is super-integrable. We
have the  second additional first integral  $\widetilde H_{3}$ which
can be chosen as a polynomial with respect to the momenta. Moreover,
there exists another additional polynomial integral of motion
$\widetilde H_4$ which can be normalised in such a way that
\[
\widetilde H_4=\{I_2,\widetilde H_3\}\,,
\]
satisfying  the equations \eqref{par-eq} as well.  For
$\kappa_2=1$ adding polynomial integrals $\widetilde H_{3,4}$ to the generators
$\rho$, $I_1$, $I_2$ we obtain the complete algebra of integrals of
motion and function $\rho$ defined  by the following brackets
\begin{equation} \label{alg-int1}
  \begin{array}{llll}
    \{I_1,\rho\}=-2\alpha I_1,&\quad\{I_1,I_2\}=0,\quad
    &\{I_1,\widetilde H_k\}=0\,,
    \quad &\{\rho,I_2\}=0\,,\\
    \\
    \{I_2, \widetilde H_3\}=\widetilde H_4\,,\quad
    &\{\widetilde H_{3,4},\rho\}=\widetilde H_{3,4},\quad
    \quad &\{I_2,\widetilde H_4\}=2\kappa^2\,I_2\widetilde H_3,\quad&\{\widetilde H_4,\widetilde H_3\}=\kappa^2\, \widetilde H_3^2\,.
  \end{array}
\end{equation}
Here parameter $\kappa$ depends on $\alpha$ and $d$. Let us present
some particular examples with cubic polynomial integrals
\begin{itemize}
 \item  for $\alpha=-\dfrac{1}{3}$, $\kappa=18$, $V=\dfrac{q_2^2}{q_1^{4}}$:
\[
\begin{split}
& \widetilde H_3=2p_2^3-2q_2p_2q_1^{-3}-p_1q_1^{-2},\\
&\widetilde H_4=12 p_2^3 (p_2 q_2-p_1 q_1)-6q_1^{-3} (p_1^2 q_1^2 + p_1 p_2 q_1 q_2 +
 4 p_2^2 q_2^2)+\frac{6 q_2^3}{q_1^6} ,\\ &
 36 H_1^3 - \widetilde H_4^2 + 36 \widetilde H_3^2 I_2=0,
\end{split}
\]
\item  for $\alpha=-1$, $\kappa=2 $, $V=\dfrac{1}{q_2^{2/3}\,q_1^{4/3}}$
\[
 \begin{split}
 & \widetilde H_3=2p_2(p_2q_2-p_1q_1)^2
-\dfrac{2p_2q_2-p_1q_1}{q_2^{2/3}q_1^{1/3}},\\
&\widetilde H_4=4 p_2 (p_2 q_2-p_1 q_1)^3-\dfrac{2 (p_1 q_1 - 4 p_2 q_2) (p_1 q_1 - p_2 q_2)}{q_1^{1/3} q_2^{2/3}}
+\frac{2}{q_1^{2/3} q_2^{1/3}},
 \end{split}
\]
\end{itemize}

fourth order integrals
\begin{itemize}
 \item  for  $\alpha=-\dfrac{1}{4}$, $\kappa=32$,  $V=\dfrac{q_2^3}{q_1^{5}}$
\[
 \begin{split}
&\widetilde H_3= p_2^4-\dfrac{3q_2^2p_2^2}{2q_1^4}-\dfrac{q_2p_2p_1}{q_1^3}-\dfrac{p_1^2}{2q_1^2}+
\dfrac{q_2^4}{4q_1^8},\\
&\widetilde H_4=8 p_2^4 (p_2 q_2-p_1 q_1)-\frac{4 (p_1^3 q_1^3 + p_1^2 p_2 q_1^2 q_2 + p_1 p_2^2 q_1 q_2^2 +
    5 p_2^3 q_2^3)}{q_1^4}+\dfrac{2 q_2^4 (3 p_1 q_1 + 5 p_2 q_2)}{q_1^8},\\
& 8 H_1^4 - \widetilde H_4^2 + 64 \widetilde H_3^2 I_2=0,
 \end{split}
\]
 \item  for $ \alpha=-1$, $\kappa=2$, $V=\dfrac{1}{q_2^{3/4}q_1^{5/4}}$
\[
 \begin{split}
 &\widetilde H_3= p_2(p_2q_2-p_1q_1)^3+\dfrac{1}{4\sqrt{q_1q_2}}-
\dfrac{(3p_2q_2-p_1q_1)(p_2q_2-p_1q_1)}{2q_2^{3/4}q_1^{1/4}},\\
 &\widetilde H_4= -2 p_2 (p_1 q_1 - p_2 q_2)^4-\dfrac{(p_1 q_1 - 5 p_2 q_2) (p_1 q_1 - p_2 q_2)^2}{q_1^{1/4}
q_2^{3/4}}+
\dfrac{3 p_1 q_1 - 5 p_2 q_2}{2 \sqrt{q_1} \sqrt{q_2}},
 \end{split}
\]

\end{itemize}
and fifth order additional integral of motion
\begin{itemize}
 \item  for $\alpha=-\dfrac{1}{5}$, $\kappa=50$,  $ V=\dfrac{q_2^4}{q_1^6}$
\[
 \begin{split}
&\widetilde H_3=\dfrac{1}{2}\,p_2\left(2p_2^2-\dfrac{q_2^3}{q_1^{5}}\right)
  \left(2p_2^2-\dfrac{3q_2^3}{q_1^{5}}\right)-\dfrac{p_1^3}{q_1^{2}}
  -\dfrac{2q_2p_2p_1^2}{q_1^{3}}
  +\dfrac{q_2^2p_1(q_2^3-3p_2^2q_1^5)}{q_1^{9}},\\
&\widetilde H_4= 20 p_2^5 (p_2 q_2-p_1 q_1)-\dfrac{10 (p_1^4 q_1^4 + p_1^3 p_2 q_1^3 q_2 + p_1^2 p_2^2 q_1^2
q_2^2 + p_1 p_2^3 q_1 q_2^3 + 6 p_2^4 q_2^4)}{q_1^5}\\
&+\dfrac{5 q_2^5 (4 p_1^2 q_1^2 + 7 p_1 p_2 q_1 q_2 + 9 p_2^2
  q_2^2)}{q_1^{10}}-\dfrac{5 q_2^{10}}{q_1^{15}}, \\
&25 H_1^5 - \widetilde H_4^2 + 100 \widetilde H_3^2 I_2=0.
 \end{split}
\]
\end{itemize}

According to \cite{ts07f,ts08b} $\widetilde H_3$ has to be function on the both
action variables $I_{1,2}$ and the single angle variable
$\omega_2$. Using the complete algebra \eqref{alg-int1} we can easily
get
\[
\widetilde H_3=\left(e^{\kappa\sqrt{2I_2}\,\omega_2}
  +c\,e^{-\kappa\sqrt{2I_2}\,\omega_2}\right)I_1^{1/2\alpha}\,.
\]
Here $\omega_2$ is defined up to canonical transformation $\omega_2\to
\omega_2+f(I_2)$ and parameter $c$ is calculated from the bracket
$\{I_2,\omega_2\}=1$. For instance, for $\alpha=-1/3$ and $d=-4$ we
have $c=-1/16$.

On the other hand,  if $\widetilde H_{3}$ is the known polynomial solution of
\eqref{par-eq}, which satisfies the algebra \eqref{alg-int1}, then we can
directly calculate $\omega_2$
\begin{equation}\label{angle-case1}
  \omega_2=\dfrac{1}{\kappa\sqrt{2I_2\,}} \, \ln\left(\,\, \dfrac{
      \widetilde H_3I_2+\sqrt{\widetilde H_3^2I_2^2-4cI_2I_1^{1/\alpha} }}{I_2
      I_1^{1/2\alpha}}\,\, \right)
\end{equation}
without separation of variables.

\section{Bi-Hamiltonian systems with higher order integrals of motion}
\label{sec:hi}
Let us recall that in Section~\ref{sec:bi} we investigated systems
given by multi-parameter family of
Hamiltonian functions given by
\begin{equation}\label{ham-monn}
  H=2p_1p_2+q_1^{-\frac{\beta(d+1)+\alpha}{\alpha}}\,q_2^d,
\end{equation}
see~\eqref{ham-mon}. For these systems with arbitrary  values of
parameters $\alpha$, $\beta$ and $d$, the corresponding bivector $P'$
given by \eqref{poi-2} is the Poisson bivector
compatible, by Proposition~\ref{pro:bi},  with canonical bivector
$P$. In previous sections we distinguished infinitely many  integrable
cases with additional first integrals which are polynomial in the momenta of the second
degree. Among them infinitely  many systems  are super-integrable.

In this section we continue our integrability analysis and we look for
cases when the additional first integral is a polynomial in momenta of
degree greater than two.
To this end we
solve equations
\eqref{par-eq} with respect to $H_2$ of the form~\eqref{anzatse}.
The existence of solutions of the obtained equations depends on
values of parameters $\alpha$, $\beta$, $d$ and
 $\kappa_2$.

Assuming that the additional first integral is of degree 4 in the
momenta and setting $\kappa_2=1$ we found   the following solutions
\begin{eqnarray}
  &&\rho=-\dfrac{1}{4}\,(q_1p_1+5p_2q_2)\,,\qquad V=q_1^3q_2^{-9/5}\,,\nn\\
  &&H_2= 4p_1^4-10\left(3p_1^2q_1^2-30p_1p_2q_1q_2+25p_2^2q_2^2\right)q_2^{-4/5}+225q_1^4q_2^{-8/5}\,,\nn\\
  \nn\\
  &&\rho=-\dfrac{1}{6}\,(2p_1q_1-p_2q_2)\,,\qquad V=q_1^2q_2^5\,,\nn\\
  &&H_2=16p_1^3(p_1q_1-p_2q_2)+4p_1q_1q_2^6(p_1q_1-2p_2q_2)+q_2^8\left(p_2^2-q_1^3q_2^4\right)\,,\nn\\
  \nn\\
  &&\rho=-\dfrac{1}{3}\,(p_1q_1+4p_2q_2)\,,\qquad V=q_1^2q_2^{-7/4}\,,\nn\\
  &&H_2=2p_1^3(p_1q_1-p_2q_2)-q_2^{-3/4}\left(13p_1^2q_1^2-80p_1p_2q_1q_2+64p_2^2q_2^2\right)+64q_1^3q_2^{-3/2}\,,\nn\\
  \nn\\
  &&\rho=-p_1q_1-\dfrac{1}{5}\,p_2q_2\,,\qquad V=q_1^{-3/4}q_2^{9/4}\,,\nn\\
  &&H_2=4p_1(p_1q_1-p_2q_2)^3-2q_1^{-3/4}q_2^{-5/4}\left(3p_1^2q_1^2-12p_1p_2q_1q_2+p_2^2q_2^2\right)
  +9q_1^{-1/2}q_2^{-5/2}\,.\nn
\end{eqnarray}
Assuming that $H_2$ is of  sixth degree with respect to the momenta  we obtained other variety
of solutions:
\begin{eqnarray}
  &&\rho=-\dfrac{1}{14}\,(7p_1q_1+p_2q_2)\,,\qquad V=q_1^{-2/3}q_2^{-10/3}\,,\nn\\
  &&H_2=4p_1^2(p_1q_1-p_2q_2)^4+(7p_1^2q_1^2-28p_1p_2q_1q_2+p_2^2q_2^2)q_1^{-4/3}q_2^{-14/3}\nn\\
  &&\,\,\,\,\,\,\,\,\,\,\,\,+2p_1(-4p_1^3q_1^3+13p_1^2p_2q_1^2q_2-20p_1p_2^2q_1q_2^2+2p_2^3q_2^3)q_1^{-2/3}q_2^{-7/3}
  -6q_1^{-1}q_2^{-7}\,,\nn\\
  \nn\\
  &&\rho=-\dfrac{1}{6}\,(p_1q_1+7p_2q_2)\,,\qquad V=q_1^2q_2^{-10/7}\,,\nn\\
  &&H_2=4p_1^6-2058q_1^3q_2^{9/7}+343q_2^{-6/7}(p_1q_1-7p_2q_2)(p_1q_1-p_2q_2)
  +14p_1^3(-4p_1q_1+7p_2q_2)q_2^{-3/7}\,,\nn\\
  \nn\\
  &&\rho=-p_1q_1+2p_2q_2\,,\qquad V=q_1^{-2/3}q_2^{-5/6}\,,\nn\\
  &&H_2=2q_1^{-2/3}\left(p_1q_1-p_2q_2\right)^2\left(
    q_2^{1/6}-2p_1q_1^{2/3}(p_1q_1-p_2q_2) \right)+q_1^{-1/3}q_2^{-1/3}\,,\nn\\
  \nn\\
  &&\rho=-\dfrac{1}{5}\,(p_1q_1+p_2q_2)\,,\qquad V=q_1^{3/2}q_2^{-7/2}\,,\nn\\
  &&H_2=q_1p_1^6-p_2q_2p_1^5-\dfrac{5q_1^{3/2}p_1^4}{2q_2^{5/2}}
  +\dfrac{3q_1^2p_1^2}{2q_2^5}
  +\dfrac{3q_1p_1p_2}{4q_2^4}+\dfrac{p_2^2}{4q_2^3}-\dfrac{q_1^{5/2}}{8q_2^{15/2}}\,.\nn
\end{eqnarray}

Here we present one super-integrable system for which equation
\eqref{par-eq} has two functionally independent solutions $H_2$ and $H_3$:
\[
\rho=-p_1q_1-\dfrac{1}{4}\,p_2q_2\,,\qquad V=q_1^{-2/3}q_2^{-7/3}\,.\nn\\
\]
In this case, for $\kappa_2=1$, there is fourth order polynomial
integral of motion

\[
H_2=\dfrac{p_1(p_1q_1-p_2q_2)^3}{2}-
\dfrac{13p_1^2q_1^2-44p_1p_2q_1q_2+4p_2^2q_2^2}{16q_1^{2/3}q_2^{4/3}}
+q_1^{-1/3}q_2^{-8/3},
\]
and for $\kappa_2=2$ sixth order polynomial first integral
\begin{eqnarray}
  H_3&=& 4p_1^2 (p_1q_1 - p_2q_2)^4 + q_1^{-4/3} q_2^{-8/3}\left(10 p_1^2q_1^2 - 16 p_1 p_2 q_1 q_2 + p_2^2 q_2^2\right)\nn\\
  &-&4q_1^{-2/3} q_2^{-4/3} p_1(p_1q_1 - p_2q_2) (2p_1^2q_1^2 -
  6p_1p_2q_1 q_2 + p_2^2 q_2^2)-\dfrac{3}{q_1 q_2^4}\,.\nn
\end{eqnarray}
In this case the complete algebra of integrals and $\rho$ is given by
\begin{eqnarray}
  &&\{\rho,H_1\}=(\alpha+\beta) H_1,
  \quad \{\rho,H_2\}=-H_2\,,\quad \{\rho,H_3\}=-2H_3\,,\quad \{\rho,H_4\}=-3H_4,\nn\\
  \nn\\
  &&\{H_2,H_3\}=H_{4},\qquad \qquad
  \{H_2,H_4\}=\dfrac{3}{8}\,H_3^2,\qquad \{H_3,H_4\}=\dfrac{27}{8}
  H_1^4\,.\nn
  \label{alg-int2}
\end{eqnarray}
\begin{remark}
  According to our knowledge it is the first example of
  super-integrable system with
   first integrals of degree two, four and six with respect to the
   momenta. In all examples known to us super-integrable systems  have
   at least two additional first integrals of degree two with respect to the
   momenta
  \cite{cd06,mw07,ran01,ts08a,ts08b}.
\end{remark}

\subsection{Additive separation of variables without
  super-integrability}
Let us consider one of the bi-Hamiltonian systems listed above
\begin{eqnarray}
  &&H_1=2p_1p_2+q_1^3q_2^{-9/5},\qquad
  \rho=-\dfrac{1}{4}\,(q_1p_1+5p_2q_2)\,,\qquad \kappa_2=1\,,\nn\\
  &&H_2=
  4p_1^4-10\left(3p_1^2q_1^2-30p_1p_2q_1q_2+25p_2^2q_2^2\right)q_2^{-4/5}+225q_1^4q_2^{-8/5}\,.\nn
\end{eqnarray}
In this case second Poisson bivector is equal to
\[
P'=\mathcal L_{\rho X} P=\dfrac{1}{2}\left(
  \begin{array}{cccc}
    0 & q_1p_1-5q_2p_2 & p_1p_2-\dfrac{3q_1^3}{2q_2^{9/5}} & {5}p_2^2+\dfrac{9q_1^4}{10q_2^{14/5}} \\
    \\
    * & 0 & p_1^2-\dfrac{15q_1^2}{2q_2^{4/5}} & {5}p_1p_2+\dfrac{9q_1^3}{2q_2^{9/5}} \\
    \\
    * & * & 0 & -\dfrac{3q_1^2}{10q_2^{14/5}}\bigl(3q_1p_1+25q_2p_2\bigr) \\
    \\
    * & * & * & 0 \\
  \end{array}
\right)
\]
and the corresponding recursion operator $\mathcal N=P'P^{-1}$ is
degenerated
\[
\det (\mathcal N-\lambda\,\id
)=\lambda^2\left(\lambda-\dfrac{3}{2}\,H_1 \right)^2.
\]
According to \cite{ts10}, we have to look for another solution $P'$ of
the equations \eqref{m-eq1} associated with the non-trivial
Darboux-Nijenhuis variables.  The existence of such solution  converts $\R^4$  into a regular bi-Hamiltonian
manifold.

It is easily to find three different second order polynomial solutions
of equations \eqref{m-eq1}.  Only one of them allows us to get
variables of separation. Namely, if
\[
Z^1=-4q_2p_1^2\,,\qquad Z^2=0\,,\qquad
Z^3=25q_2^{1/5}(3p_1q_1+5p_2q_2)\,,\] and
\[Z^4=-5q_2^{-6/5}
\left((2p_1p_2q_2^{9/5}+q_1^3)\sqrt{10}-3p_1q_2^{2/5}q_1^2-15p_2q_2^{7/5}q_1\right)
\]
are coordinates of the vector field $Z=\sum Z^k\partial_k$, then
bivector
\[
P'=\mathcal L_Z P
\]
is the solution of \eqref{m-eq1} for the given integrable system.  For
this Poisson bivector $P'$ the control matrix $F$
\begin{equation}\label{f-mat}
  P'\mathrm{d}H_i=P\sum_{j=1}^2 F_{ij}\,\mathrm{d}H_j,\qquad
  i=1, 2,
\end{equation}
is a non-degenerate matrix of the form
\[
F=\left(
  \begin{array}{cc}
    10q_2^{1/5}(-15q_1+q_2^{2/5}\sqrt{10}p_1) & \dfrac{1}{2} \\
    -500q_2^{1/5}(2q_2p_1^2+15q_2^{1/5}q_1^2-3q_2^{3/5}\sqrt{10}p_1q_1+5q_2^{8/5}\sqrt{10}p_2) & 0 \\
  \end{array}
\right)\,.
\]
The eigenvalues $u_{1,2}$ of  $F$
\begin{eqnarray}
  &&A(\lambda)=(\lambda-u_1)(\lambda-u_2)=\nn\\
  &&=\lambda^2-10q_2^{1/5}(q_2^{2/5}\sqrt{10}p_1-15q_1)\lambda+250q_2^{2/5}
  \Bigl(15q_1^2+\sqrt{10}q_2^{2/5}(5p_2q_2-3q_1p_1)+2p_1^2q_2^{4/5}\Bigr),
  \nn
\end{eqnarray}
coincide with eigenvalues of the corresponding recursion operator $
\mathcal N=P'\cdot P^{-1}=\mathcal L_Z P\cdot P^{-1} $.  Thus
$u_{1,2}$ are the variables of separation and the so-called
Darboux-Nijenhuis coordinates. The conjugated momenta $v_{1,2}$ are
equal to
\[
v_{1,2}=-\dfrac{q_1}{40\sqrt{10}q_2^{3/5}}-\dfrac{p_1}{200\,q_2^{1/5}}\pm
\dfrac{u_1-u_2}{200\,q_2^{4/5}},\qquad \{u_j,v_k\}=\delta_{ij}\,.
\]
The separated equations look like
\begin{equation}\label{sep-eqUV}
  -2\,\left(\sigma
    v_k\right)^3 u_k+H_2+2u_kH_1=0\,,\qquad k=1,2,\quad
  \sigma=10^{3/2}\,.
\end{equation}
Thus, in order to get solutions $v_k(t,\beta_1,\beta_2)$ and
$u_k(t,\beta_1,\beta_2)$ we have to solve the following Abel-Jacobi
equations
\begin{eqnarray}
  \beta_1-t&=&H_2\int^{v_1}\dfrac{\mathrm{d}v}{(\sigma^3 v^3+H_1)^2}+H_2\int^{v_2}\dfrac{\mathrm{d}v}{(\sigma^3 v^3+H_1)^2}\,,\nn\\
  \beta_2&=&\int^{v_1}\dfrac{\mathrm{d}v}{2(\sigma^3v^3-H_1)}+\int^{v_2}\dfrac{\mathrm{d}v}{2(\sigma^3v^3-H_1)}.
\end{eqnarray}
There are no addition theorem for the above quadratures and,
therefore, we can suppose, that there are no additional polynomial
integrals of motion \cite{ts08a,ts08b}. In fact, the second angle
variable
\begin{eqnarray}
  \omega_2&=&\dfrac{1}{2}\sum_{k=1}^2\int^{v_k}\dfrac{\mathrm{d}v}{\sigma^3v^3-H_1}
  =\dfrac{1}{12\sigma H_1^{2/3}}
  \sum_{k=1}^2\left[\,2\ln\left(\sigma v_k-H_1^{1/3}\right)\right.\nn\\
  &-& \left.\ln\left(\sigma^2v_k^2+\sigma v_kH_1^{1/3}+H_1^{2/3}\right)
    -2\sqrt{3}\arctan\left(\dfrac{2\sigma
        v_k+H_1^{1/3}}{\sqrt{3}H_1^{1/3}}\right)\right] \nn
\end{eqnarray}
cannot be rewritten as a function on polynomials $H_1, H_2$ and $H_3$
as in the super-integrable case \eqref{angle-case1}.

In the similar manner we can get solutions of the equations of motion
for other systems with fourth order integral of motion. Construction
of the variables of separation for the systems with sixth order
integrals of motion is an open problem.

\section*{Acknowledgements}
AJM and MP  thank  very much Artur Sergyeyev and Maciej B\l{}aszak for  discussions related to this work.
For them this research has been partially supported by grant No. N N202 2126 33
of Ministry of Science and Higher Education of Poland. For AJM this
research has been also  partially supported by EU
funding for the Marie-Curie Research Training Network AstroNet.

\appendix
\section{Appendix. Hypergeometric equation}
\label{sec:app}
Hypergeometric equation
\begin{equation}
  z(1-z)w''+[c-(a+b+1)z]w'-abw=0,
  \label{eq:riemann}
\end{equation}
is a special case of the Riemann $P$ equation. Thus it possesses three
regular singularities at $z=0$, $z=1$ and $z=\infty$, with exponents
and their differences of the following form
\begin{align*}
  &z=0, && \rho_1=1-c, && \rho_2=0,&& \lambda=1-c,\\
  &z=1,&&\sigma_1=c-a-b,&&\sigma_2=0,&&\nu=c-a-b,\\
  &z=\infty,&&\tau_1=b,&&\tau_2=a,&&\mu=b-a.
\end{align*}
Exponents satisfy the Fuchs relation, i.e. their sum is equal to one.

As it is well known one of  solutions of equation~\eqref{eq:riemann}
holomorphic at the origin is the  hypergeometric function
given by the following series
\begin{equation}
\label{eq:ghyp}
 \phantom{\vert}_2F_{1}\left(a,b,c, z
  \right):= \sum_{k=0}^{\infty}
\dfrac{(a)_k (b)_k }{(c)_k k!} z^{k}.
\end{equation}
Here $(x)_k=x(x+1)\cdots(x+k)$ is the Pochhammer symbol and $(x)_0=1$.
The above series is a polynomial iff $a$ or $b$ is a non-positive
integer, see, e.g., \cite{Poole:60::,Iwasaki:91::}.

The following lemma gives the necessary and sufficient condition for
\eqref{eq:riemann} to be reducible. It is a classical, well known
fact, see, e.g., \cite{Iwasaki:91::}.
\begin{lemma}
  \label{lem:redrie}
  Equation~\eqref{eq:riemann} is reducible if and only there exist
  $i$, $j$, $k\in \{1,2\}$, such that
  \begin{equation}
    \rho_i +\sigma_j +\tau_k \in \Z.
  \end{equation}
  Equivalently, equation~\eqref{eq:riemann} is reducible if and only
  if at least one number among
  \begin{equation}
    \lambda+\mu+\nu,\quad
    -\lambda+\mu+\nu,\quad  \lambda-\mu+\nu,\quad  \lambda+\mu-\nu,
  \end{equation}
  is an odd integer.
\end{lemma}
From the above lemma it follows that if equation~\eqref{eq:riemann} is
reducible, then we can always renumber exponents in such a way that
\begin{equation*}
 \rho_1+\sigma_1+\tau_1\in-\N_0,
\end{equation*}
where $\N_0$ denotes the set of non-negative integers. But then, from
the Fuchs relation, we also have
\begin{equation*}
  \rho_2+\sigma_2+\tau_2\in\N.
\end{equation*}
Hence, if~\eqref{eq:riemann} is reducible, we assume from now on that
the exponents are numbered in this way.

If the difference of exponents at a singular point is an integer, then
it can happen that a local solution around this singularity contains a
logarithm. Such a singularity is called logarithmic. For
equation~\eqref{eq:riemann}, it is enough to know the exponents to
decide which singularity is logarithmic. To formulate the next lemma
which gives the necessary and sufficient conditions for a singularity
of \eqref{eq:riemann} to be logarithmic we introduce the following
notation. For a non-negative integer $m\in\N_0$ we define
\begin{equation}
  \langle m\rangle:=\begin{cases}
    \emptyset & \mtext{if} m=0,\\
    \{1, \ldots, m\} &\mtext{otherwise.}
  \end{cases}
  \label{eq:mym}
\end{equation}
For $s\in\{0,1,\infty\}$ let $e_{s,1}$ and $e_{s,2}$ denote exponents
of equation~\eqref{eq:riemann}, ordered in such a way that $\Re
e_{s,1}\geq \Re e_{s,2}$. With the above notation we have the
following.
\begin{lemma}
  \label{lem:iwa3}
  Let $r\in\{0,1,\infty\}$. Then $r$ is a logarithmic singularity of
  equation~\eqref{eq:riemann} if and only if $m:=e_{r,1}-e_{r,2}\in
  \N_0$, and
  \begin{equation}
    e_{r,1}+e_{s,i} +e_{t,j}\not\in \langle m\rangle, \mtext{for} i,j\in\{1,2\},
    \label{eq:sumiaki}
  \end{equation}
  where $r,s,t$ are pairwise different elements of $\{0,1,\infty\}$.
\end{lemma}
For more details, see Lemma~4.7 and its proof on pp. 91--93
in~\cite{Iwasaki:91::}.

For considerations of this paper we need characterisation when
reducible hypergeometric equation has algebraic solutions.  The answer
is contained in the following theorem.
\begin{theorem}
  \label{thm:kimura}
  Assume that hypergeometric equation \eqref{eq:riemann} is reducible.
  Then all its solutions are algebraic iff
  \begin{enumerate}
  \item all exponents are rational, and
  \item exactly two or four of $\lambda+\mu+\nu$, $-\lambda+\mu+\nu$,
    $\lambda-\mu+\nu$, $\lambda+\mu-\nu$ are odd integers, and
  \item no one of singularities is logarithmic.
  \end{enumerate}
\end{theorem}
For proof and details see \cite{Kimura:69::,Iwasaki:91::,Poole:60::}.


\begin{thebibliography}{10}

\bibitem{Audin:01::c} M. Audin.  \newblock \emph{Les syst\`emes
    hamiltoniens et leur int\'egrabilit\'e}.  \newblock Cours
  Sp\'ecialis\'es 8, Collection SMF. SMF et EDP Sciences, Paris, 2001.

\bibitem{05.0470.01} J. Bertrand, Th\'eor\`eme relatif au mouvement
  d'un point attir\'e vers un centre fixe.  \newblock
  \emph{C. R. Acad. Sci. Paris} v. LXXVII, pp. 849--853, 1873.

\bibitem{Bogoyavlenskii:98::a}
O. I. Bogoyavlenskij.
\newblock Algebraic properties of master symmetries and their
applications.
\newblock\emph{Dokl. Akad. Nauk\/}, v. 360(4), pp. 445--447, 1998.

\bibitem{Borisov:09::} A.~V. Borisov, A. A.~Kilin, and I.~S. Mamaev.
  \newblock Superintegrable {S}ystems on a {S}phere with {I}ntegral of
  {H}igher {D}egree.  \newblock \emph{Regul. Chaotic Dyn.\/}, v.14
  (6), pp. 615--620, 2009.

\bibitem{cd06} C. Daskaloyannis, K. Ypsilantis, \newblock Unified
  treatment and classification of superintegrable systems with
  integrals quadratic in momenta on a two-dimensional manifold,
  \emph{J. Math. Phys.\/}, v.47, 2006, p. 042904, 38 pages,
  arXiv:math-ph/0412055v3.

\bibitem{ts05} Yu.A. Grigoryev, A.V. Tsiganov, \newblock Symbolic
  software for separation of variables in the Hamilton-Jacobi equation
  for the L-systems, \emph{Regul. Chaotic Dyn.\/}, v. 10(4),
  pp. 413--422, 2005.

\bibitem{Iwasaki:91::} K. Iwasaki, H. Kimura, S. Shimomura, and
  M. Yoshida, {\em From {G}auss to {P}ainlev\'e, {A} modern theory of
    special functions}, Aspects of Mathematics, E16.  \newblock
  Friedr. Vieweg \& Sohn, Braunschweig, 1991.

\bibitem{Kalnins:09} E. G. Kalnins, W. Miller Jr., G. S. Pogosyan,
  \newblock Superintegrability and higher order constants for
  classical and quantum systems, arXiv:0912.2278v1 [math-ph], 24
  pages.

\bibitem{Kimura:69::} T. Kimura, On Riemann's equations which are
  solvable by quadratures.  \newblock \emph{ Funkcial. Ekvac.}, v. 12,
  pp. 269--281, 1969/1970.

\bibitem{Maciejewski:04::g} A.~J. Maciejewski, M. Przybylska, All
  meromorphically integrable 2{D} Hamiltonian systems with homogeneous
  potential of degree 3, \emph{Phys.  Lett. A}, v. 327(5-6),
  pp. 461--473, 2004.

\bibitem{mp:05::c} A.~J.Maciejewski, M. Przybylska, Darboux points and
  integrability of {H}amiltonian systems with homogeneous polynomial
  potential, \emph{J. Math.  Phys.}, v. 46(6), p. 062901, 33 pages,
  2005.

\bibitem{mp:08::c} A. J. Maciejewski, M. Przybylska, H. Yoshida,
  \newblock Necessary conditions for super-integrability of
  Hamiltonian systems.  \newblock \emph{ Phys. Lett. A}, v. 372(34),
  pp. 5581--5587, 2008.


\bibitem{amp09} A. J. Maciejewski, M. Przybylska, \newblock
  Differential Galois theory and Integrability.  \newblock
  \emph{Internat. J.~Geom. Methods in Modern Phys.}, v. 6(8),
  pp. 1357--1390, 2009.


\bibitem{mw07} I. Marquette, P. Winternitz, \newblock Polynomial
  Poisson algebras for classical superintegrable systems with a
  third-order integral of motion, \emph{J.~Math.~Phys.\/} v.48,
  p. 012902, 2007.


\bibitem{Morales:99::c} J.~J. Morales~Ruiz.  \newblock
  \emph{Differential {G}alois Theory and Non-Integrability of
    {H}amiltonian Systems}, volume 179 of \emph{Progress in
    Mathematics\/}.  \newblock Birkh\"auser Verlag, Basel, 1999.

\bibitem{Morales:01::a} J.~J.  Morales-Ruiz, J.~P. Ramis, A note on
  the non-integrability of some {H}amiltonian systems with a
  homogeneous potential, \emph{Methods Appl.  Anal.\/}, v. 8(1),
  pp. 113--120, 2001.


\bibitem{Nakagawa:02::} K. Nakagawa, \emph{Direct construction of
    polynomial first integrals for {H}amiltonian systems with a
    two-dimensional homogeneous polynomial potential\/}, Ph.D. thesis,
  The Graduate University for Advanced Studies, Japan, 2002.


\bibitem{Poole:60::} E. G. C. Poole, \emph{Introduction to the theory
    of linear differential equations\/}.  \newblock Dover Publications
  Inc., New York, 1960

\bibitem{mp09} M. Przybylska, \newblock{ Darboux points and
    integrability of homogenous Hamiltonian systems with three and
    more degrees of freedom.}  \newblock \emph{Regul. Chaotic Dyn.},
  v. 14(2), pp. 263--311, 2009.

\bibitem{mp09a} M. Przybylska, \newblock{ Darboux points and
    integrability of homogenous Hamiltonian systems with three and
    more degrees of freedom. Nongeneric cases.}  \newblock
  \emph{Regul. Chaotic Dyn.\/}, v. 14(3), pp. 349--388, 2009.

\bibitem{ran01} M.F. Ranada, M. Santander, \newblock Complex Euclidean
  super-integrable potentials, potentials of Drach, and potential of
  Holt, \emph{ Phys. Lett. A\/}, v. 278, pp. 271--279, 2001.


\bibitem{ts06} A. V.  Tsiganov, \newblock{ Towards a classification of
    natural integrable systems}, \emph{Regul. Chaotic Dyn.\/},
  v. 11(3), pp. 343--362, 2006.

\bibitem{ts07f} A.V. Tsiganov, \newblock{ On maximally superintegrable
    systems}, \emph{Regul. Chaotic Dyn.\/}, v. 13, pp. 178--190, 2008.

\bibitem{ts08a} A.V. Tsiganov, \newblock{ Addition theorem and the
    Drach superintegrable systems}, \emph{J. Phys. A: Math. Theor.\/},
  v. 41(33), p. 335204 (16pp), 2008.

\bibitem{ts08b} A.V. Tsiganov, \newblock{ Leonard Euler: addition
    theorems and superintegrable systems}, \emph{Regul. Chaotic
    Dyn.\/}, v. 14(3), pp. 389--406, 2009.

\bibitem{ts10} A.V. Tsiganov, \newblock{On bi-integrable natural
    Hamiltonian systems on the Riemannian manifolds}, arXiv:1006.3914, accepted to Journal of Nonlinear Mathematical Physics,   2010.




\end{thebibliography}
\end{document}